\documentclass[11pt]{article}
\usepackage[margin=1in]{geometry}
\usepackage[ruled]{algorithm2e} % For algorithms

\usepackage{amsmath}
\usepackage{amsthm}
\usepackage{bbm}
\usepackage{amsfonts}
\usepackage{xcolor}
\usepackage{cancel}
\usepackage{hyperref}
\usepackage{graphicx}
\usepackage{mathtools}
\usepackage{thmtools} 
\usepackage{natbib}
\usepackage{bm}
\usepackage{float}

\graphicspath{{figs/}}

\usepackage{color}
\definecolor{darkblue}{rgb}{0.0,0.0,0.2}

\newenvironment{proofsketch}{%
  \proof}{\endproof}

\DeclarePairedDelimiterX{\infdivx}[2]{(}{)}{%
  #1\;\delimsize\|\;#2%
}
\newcommand{\kldiv}{D_{\text{KL}}\infdivx}

\newcommand{\reals}{\mathbb{R}}
\newcommand{\E}{\mathop{\mathbb{E}}}
\newcommand{\ones}{\mathbbm{1}}

\newcommand{\X}{\mathcal{X}}
\newcommand{\Y}{\mathcal{Y}}

\newcommand{\countv}{k_{\bm{x}}}
\newcommand{\emp}{\hat{P}_{\bm{x}}}
\newcommand{\game}{\mathcal{G}_n(\X, \Y, d, Q)}
\newcommand{\pmin}{p_0}
\DeclareMathOperator\supp{supp}

\newcommand{\sigmaprofile}{\boldsymbol{\sigma}}

\newtheorem{lemma}{Lemma}
\newtheorem*{informaltheorem}{Theorem (informal)}
\newtheorem{corollary}{Corollary}
\newtheorem{theorem}{Theorem}
\newtheorem*{theorem*}{Theorem}
\newtheorem{proposition}{Proposition}

\theoremstyle{definition}
\newtheorem{observation}{Observation}

\theoremstyle{definition}
\newtheorem{condition}{Condition}

\theoremstyle{definition}
\newtheorem{definition}{Definition}

\usepackage[colorinlistoftodos,textsize=tiny]{todonotes}
\usetikzlibrary{shapes, calc, arrows.meta, positioning}
\tikzset{notestyleraw/.append style={inner sep = 2pt}}

\begin{document}

\title{Equilibria in Large Position-Optimization Games}

\author{Rafael Frongillo, Melody Hsu, Mary Monroe, and Anish Thilagar \\
University of Colorado Boulder}
\date{}
\maketitle

\begin{abstract}
    We propose a general class of symmetric games called position-optimization games.
    Given a probability distribution $Q$ over a set of targets $\Y$, the $n$ players each choose a position in a space $\X$.
    A player's utility is the $Q$-mass of targets they are closest to under some proximity measure, with ties broken evenly.
    Our model captures Hotelling games and forecasting competitions, among other applications.
    We show that for sufficiently large $n$, both pure and symmetric mixed Nash equilibria exist, and moreover are \emph{extreme}: all players play on a finite set of \emph{pseudo-targets} $\X^* \subseteq \X$.
    We further show that both pure and symmetric mixed equilibria converge to the distribution $P$ on $\X^*$ induced by $Q$, and bound the convergence rate in $n$.
    The generality of our model allows us to extend and strengthen previous work in Hotelling games, and prove entirely new results in forecasting competitions and other applications.
\end{abstract}

\section{Introduction} \label{sec:intro}

Consider the following abstract, non-cooperative game: $n$ agents compete by playing over a common set of positions in some space $\X$, with the goal of capturing targets in some space $\Y$ distributed according to some measure $Q$.
An agent captures a target if their position minimizes some measure of proximity to it (with ties broken randomly).
We call this setting a \emph{position-optimization game}.
The model captures a variety of previously studied games, including versions of Hotelling's co-location model, forecasting competitions, and Voronoi games (\S~\ref{subsec:intro-examples}).

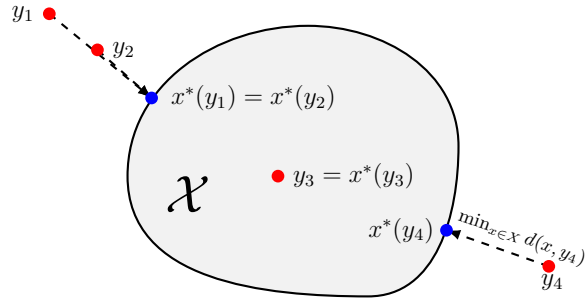
\begin{figure}[H]
\centering
\begin{minipage}[t]{0.9\linewidth}  % left half
\centering
\begin{tikzpicture}[>=latex, scale=0.8, transform shape]

    \def\blobX{(0,0) to[out=90,in=180] (3,2.5) to[out=0,in=90] (5.5,0.5) to[out=270,in=0] (4,-2) to[out=180,in=270] (0,0)}

    \fill[gray!10] \blobX;
    \draw[thick, color=black] \blobX;
    \node[font=\Huge, color=black] at (1, -0.3) {$\mathcal{X}$};

    % --- Define Points ---

    \coordinate (y_red) at (2.5, 0);
    
    % 2. The Blue Point (OUTSIDE X - Top Left)
    \coordinate (y_blue) at (-1.3, 2.7);
    \coordinate (x_star_blue) at (0.4, 1.3); % Visual projection onto boundary
    \coordinate (y_orange) at (-0.5, 2.1);
    
    % 3. The Green Point (OUTSIDE X - Bottom Right)
    \coordinate (y_green) at (7, -1.5);
    \coordinate (x_star_green) at (5.3, -0.9); % Visual projection onto boundary

    % Line for Blue
    \draw[dashed, black, thick, ->] (y_blue) to (x_star_blue);

    \draw[dashed, black, thick, ->] (y_orange) -- (x_star_blue);
    
    % Line for Green
    \draw[dashed, black, thick, ->] (y_green) -- (x_star_green)  node[midway, above, sloped, color=black, font=\footnotesize] {$\qquad \min_{x \in X} d(x,y_4)$};

    % --- Draw and Label Points ---

    \fill[red] (y_red) circle (3pt) node[right, black, font=\large] {$\; y_3 = x^*(y_3)$};

    \fill[red] (y_orange) circle (3pt) node[right, black, font=\large] {$\; y_2$};

    \fill[red] (y_blue) circle (3pt) node[left, black, font=\large] {$y_1 \;$};
    \fill[blue] (x_star_blue) circle (3pt) node[right, black, font=\large, xshift=2pt] {$\; x^*(y_1) = x^*(y_2)$};

    \fill[red] (y_green) circle (3pt) node[below, black, font=\large] {$\; y_4$};
    \fill[blue] (x_star_green) circle (3pt) node[left, black, font=\large, xshift=-2pt] {$x^*(y_4)$};

\end{tikzpicture}
\caption{An example of spaces $\X$ and $\Y = \{y_1, y_2, y_3, y_4\}$, here both subsets of some Euclidean space with proximity function $d$.
Note multiple targets ($y_1$, $y_2$) can map to the same pseudo-target, and targets in $\X$ are themselves pseudo-targets (see $y_3$).}
\label{fig:xstar-y}
\end{minipage}
\end{figure}

These games and their many variations have been studied extensively, but mainly under limited conditions---for specific choices of $\X$, $\Y$, and $Q$, or under strong structural constraints on agent strategies.
For example, the finite-location Hotelling game of \citet{nunez2016competing,nunez2017large} involves the special case where $\X$ is a finite set of retailers competing for the attention of consumers $\Y$.
In forecasting competitions, our model captures the special case where all agents have the same information over the outcome, which perhaps surprisingly is already a nontrivial game;
in the sole work addressing this forecasting game, \citet{laster1999rational} assume the $n$ agents form a continuum, i.e., they study the continuous limit as $n\to\infty$.
Similar to other variations of the Hotelling model (\S~\ref{sec:related-work}), the conclusions of these works is that for sufficiently large $n$, the equilibria approach a distribution determined by $Q$.
For example, in a Hotelling game, retailers will follow the distribution of consumer preferences over positions when they set up shop; in forecasting competitions, forecaster reports will mimic the underlying distribution over outcomes. 

Despite using similar techniques, the above results all appear in disjoint models---sometimes only loosely related, as in Hotelling versus forecasting---with no clear technical connection between them.
Moreover, several important questions remain unanswered.
Most glaringly, there is no existing work on general equilibria of forecasting competitions for any finite value of $n$.
The finite-location Hotelling model does give results for sufficiently large $n$, but without convergence rates in $n$, which are crucial to understand when one can rely on properties of the limiting behavior in practice.

In this paper, we give general equilibrium results and convergence rates for position-optimization games for large $n$, illuminating that these disjoint settings are all special cases of the same fundamental structure.
We show the existence of symmetric mixed Nash equilibria for all $n$, and pure equilibria for sufficiently large $n$.
Furthermore, we show that both pure and mixed equilibria converge at a rate of $O(1/n)$ to $P$, the projection of $Q$ onto the \emph{pseudo-targets} $\X^*$. (We define $\X^* \subseteq \X$ as the positions that are optimal for some $y\in\Y$; see Figure~\ref{fig:xstar-y}).
For example, in the finite-location Hotelling game of \citet{nunez2016competing}, $\X^*$ is the set of all first-ranked locations of consumers; in forecasting competitions, $\X^*$ is the set of forecasts giving probability 1 or 0 to every event.)

\begin{informaltheorem}
    Assume $\X^*$ is a finite set.
    For sufficiently large $n$, both pure and symmetric mixed equilibrium strategies are supported on $\X^*$.
    Moreover, the distribution over $\X^*$ induced by pure equilibria (via the empirical distribution of agent positions) and by symmetric mixed equilibria (via the underlying mixed strategy) both converge to $P$ at a rate of $O(1/n)$.
\end{informaltheorem}

Our results show that, even in this general model, agents continue to mimic the underlying distribution of targets through their chosen positions for large $n$.
As discussed further in \S~\ref{sec:related-work}, our result is the first to characterize convergence rates of the game, and the first to characterize general forecasting competition equilibria.
We also extend previous asymptotic results of the original Hotelling game to non-differentiable (specifically, finite-support) consumer distributions, and strengthen the assumptions and results for finite-location Hotelling games.
We suspect that our work implies novel results for other unexplored instances of this general game as well.

\subsection{Motivating examples} \label{subsec:intro-examples}
    Position-optimization games arise across many domains, including network science and voting theory. 
    We outline several below, and later reference these settings to illustrate our general results. 
    
    \begin{enumerate}
        \item \emph{Forecasting competitions.}
        In forecasting competitions, forecasters compete to make the best predictions of $m$ random events.
        One example is the machine learning competition platform Kaggle, where forecasters are provided with the same set of training data and compete to develop the best prediction models.
        These competitions typically score each forecaster according to a strictly proper scoring rule, which is known to be truthful in isolation.
        However, others have observed that the winner-takes-all nature of the competition causes strategic behavior in the simplest one-event case~\citep{lichtendahl2007probability}.
        This tension has motivated extensive game-theoretic study of such competitions \citep{frongillo2021efficient,witkowski2023incentive,monroe2025hedging}.
        Under the position-optimization game framework, we observe strategic non-truthful reporting, even when agents share a common-knowledge belief $Q$.
        
        Formally, we model the random events as a vector $Y \in \{0,1\}^m$ with distribution $Q \in \Delta\left(\{0,1\}^m\right)$.
        Each forecaster $i$ submits a vector of marginals for each event $x_i \in [0,1]^m =: \X$.
        The winner is the forecaster whose prediction has the highest accuracy as judged e.g.\ by the Brier scoring rule, which corresponds to the $x_i$ closest to $Y$ in $\ell_2$-distance (ties broken at random).
        We visualize this model of forecasting competitions for $m=3$ in Figure~\ref{fig:forecasting}.
        
        \item \emph{Finite-location Hotelling games.}
        Hotelling games are a well-studied class of spatial games involving a set of \emph{consumers} and \emph{retailers.}
        Retailers choose where to set up their store locations in some space, in order to attract the largest mass of consumers under some distribution.
        Consider the finite-location model introduced by \citet{nunez2017large}.
        Let $\mathcal{S}$ be a compact subset of a Euclidean space with distance measure $d$.
        Then retailers choose locations $\bm{x} = (x_1, \ldots, x_n)$ from some finite set $\X = \{1, \ldots, k\} \subset \mathcal{S}$.
        Consumers in $\mathcal{S}$ are distributed according to an absolutely continuous distribution $\lambda$ over $\mathcal{S}$; we need not assume anything about the cardinality of consumers. 
        Each retailer's utility is then the share of customers who patronize their shops, given each customer will visit the location that is preferred most in $\bm{x}$ under distance $d$.
        We give a visualization of this game in Figure~\ref{fig:finite-hotelling}.

        \item \emph{Classic Hotelling games.}
        The classic Hotelling game allows retailers to choose any location on the unit interval, unlike the finite-location version. 
        Consumers follow some distribution on the same interval.
        
        \item \emph{Spatial voting.}
        Closely related to Hotelling games is the model of spatial voting (see e.g. \citet{merrill1999unified}), often explored in the literature on metric distortion
        \citep{anshelevich2018approximating}. 
        We adapt the model in this literature to our setting, where candidates can strategize over some underlying ideological space and the winner is chosen according to plurality rule.
        
        \item \emph{Discrete Voronoi games.}
        In one-round discrete Voronoi games \citep{durr2007nash,sun2020voronoi}, a set of strategic influencers place themselves on vertices in a graph in order to attract the most users as measured by shortest path distance. 
    \end{enumerate}

\begin{figure}[htbp]
\centering
\begin{minipage}[t]{0.48\linewidth}
\centering
\includegraphics[scale=0.23]{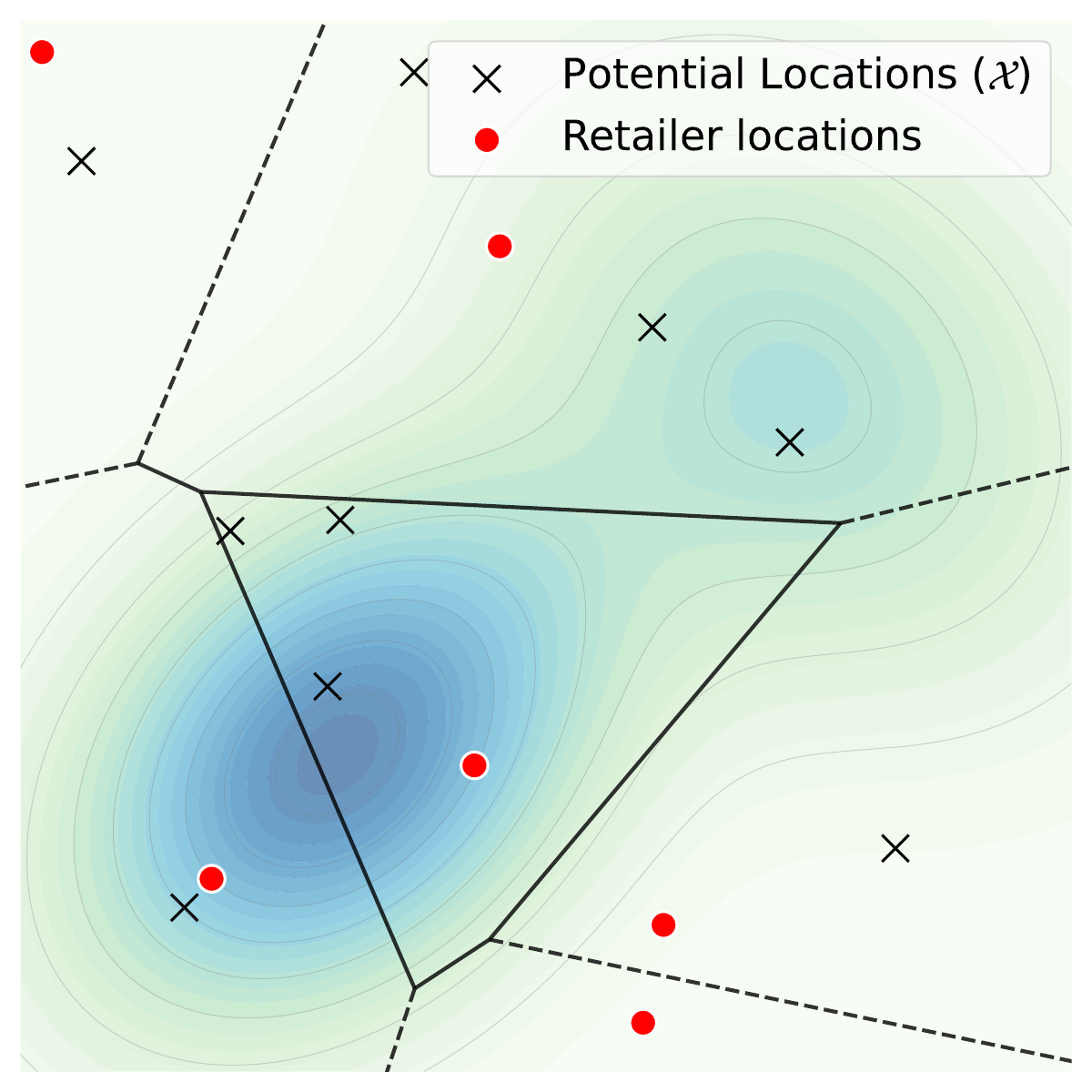}

\caption{The finite-location Hotelling game, where $\X$ is a finite set of locations in a Euclidean space.  Retailers choose locations from $\X$ (red). Consumers follow a continuous distribution, and retailers win mass inside their Voronoi cell relative to the other retailer locations.}
\label{fig:finite-hotelling}

\end{minipage}
\hfill
% ---------- RIGHT COLUMN: placeholder for another figure ----------
\begin{minipage}[t]{0.48\linewidth}
\centering
%\framebox[0.95\linewidth]{\parbox[c][3cm][c]{0.9\linewidth}{\centering forecasting example here}}
\resizebox{0.6\linewidth}{!}{%
\begin{tikzpicture}[scale=4]
% --- cube vertices ---
\coordinate (A) at (0,0,0);
\coordinate (B) at (1,0,0);
\coordinate (C) at (1,1,0);
\coordinate (D) at (0,1,0);

\coordinate (E) at (0,0,1);
\coordinate (F) at (1,0,1);
\coordinate (G) at (1,1,1);
\coordinate (H) at (0,1,1);

% --- cube edges ---
\draw (A)--(B)--(C)--(D)--cycle;
\draw (E)--(F)--(G)--(H)--cycle;
\draw (A)--(E);
\draw (B)--(F);
\draw (C)--(G);
\draw (D)--(H);

% --- black dots on vertices ---
\foreach \v in {A,B,C,D,E,F,G,H} {
  \fill[black] (\v) circle (0.02);
}

% --- upper-right-front vertex ---
\coordinate (G) at (1,1,1);

% --- gray sphere centered at G ---
\shade[ball color=gray!40, opacity=0.5] (G) circle (0.28);

% --- blue point at the same vertex ---
\fill[blue] (G) circle (0.03);

% --- red points ---
\fill[red] (0.2,0.3,0.15) circle (0.03);
\fill[red] (0.25,0.45,0.3) circle (0.03);
\fill[red] (0.5,0.2,0.25) circle (0.03);
\fill[red] (0.4,0.1,0.6) circle (0.03);
\fill[red] (0.72,0.83,0.9) circle (0.03);

% --- dashed segment ---
\draw[dashed, thick] (0.75,0.85,0.9) -- (G);

\end{tikzpicture}
}
\caption{A forecasting competition with $m = 3$ binary events.  Targets correspond to vertices of the cube $\{0,1\}^3$.  Forecasters submit predictions (red) in the filled cube $[0,1]^3$.  The outcome (blue) is drawn from distribution $Q$ over vertices, with the winner minimizing $\ell_2$ distance.}
\label{fig:forecasting}
\end{minipage}
\label{fig:shops-and-agents-left-right}
\end{figure}

    \subsection{Related work}
    \label{sec:related-work}
    
    The majority of related work studies variations of the classic Hotelling game, where firms compete for consumers by location (and not on price), with the firm capturing the set of consumers closest to them.%
    \footnote{This version is called the \emph{pure} Hotelling game is the literature. In this paper, we always mean the pure Hotelling game, but omit the word "pure" to avoid confusion with pure and mixed strategies.  We refer the reader to \cite{DREZNER20245} for a more detailed survey of spatial competition.}
    \citet{osborne1986nature} study the game on the unit interval where firms employ mixed strategies and consumer distributions are nonuniform (but still twice-differentiable).
    In the case where the CDF of consumers $F$ is twice-differentiable, the authors show that the distribution of retailers asymptotically approaches $F$ over $n$. 
    Our paper is the first to derive results under the Hotelling model for non-differentiable $F$, in particular for the case where the support of the consumer distribution is finite (and the position set is still infinite).

    \citet{nunez2016competing} study an instantiation of our game where the consumer space is abstracted away: each consumer is identified with a strict ranking over locations in finite set $\X$.
    Their results show that pure equilibria exist, and the empirical distribution in pure equilibria converges to the true underlying mass of first-place consumer rankings over the finite set of locations.
    However, they do not provide rates of convergence in $n$; and while their equilibrium construction is only shown to terminate in finite time, our simpler construction admits an explicit runtime bound of $O(|\X^*|n).$
    The authors also assume preferences are strict, whereas we only require that consumers have a unique first preference. 
    Moreover, the authors assume no matter where retailers are located, each receives at least one consumer. 
    We do not need this assumption, but instead derive it in our characterization of equilibria.
    
    A follow-up paper~\citep{nunez2017large} derives similar results for symmetric mixed equilibria.
    Here, the authors study a specific case of the general finite-location setting in their previous paper~\citep{nunez2016competing}: the consumer set $\Y$ is a compact subset of a Euclidean space, and retailers can choose from a finite set of locations in the same space.
    Given the consumer distribution over target set $\Y$ is absolutely continuous, the authors prove an analagous result that the symmetric mixed equilibria converges to the true underlying mass of consumer preferences.
    Again, they do not provide rates of convergence in $n$.
    Our mixed equilibrium results are more general, as we do not impose Euclidean structure on the sets $\X$ and $\Y$, nor do we assume that the set $\X$ is finite, a case crucial for forecasting competitions. 
    
    In the forecasting competition literature, 
    \citet{laster1999rational} study forecasting competitions with one event ($\Y = \{0,1\}$), and assume a continuum of agents.
    \citet{ottaviani2006strategy} study settings where the event is continuous with a Normal distribution $Q$.
    Both papers observe in their specific cases that symmetric mixed equilibria will converge to the event distribution for infinitely many players; neither give convergence rates. 
    \citet{lichtendahl2007probability} observe that truthfulness is not an equilibrium for the one-event setting.
    Given this fact, all further work in the area either restricts to two-player settings with specific instantiations of $Q$ \citep{monroe2025hedging}, or pivots to proposing new competition mechanisms~\citep{frongillo2021efficient,witkowski2023incentive}.
    Thus our paper is the first to explicitly characterize general equilibria in traditional forecasting competitions.
    
\section{Setting}\label{sec:setting}
\subsection{Model}
Let $\X$ be an arbitrary action set; we call $x\in\X$ a \emph{position}.
Let $\Y$ be an arbitrary target set; we call $y \in \Y$ a \emph{target}.
We place no assumptions of finiteness on $\X$ or $\Y$, nor do we assume that they coincide or are subsets of the same space. 
Let $d:\X\times\Y\rightarrow[0,\infty]$ be a \emph{proximity function} over positions and targets.
Let $x^*(y)=\arg\min_{x\in\X}d(x,y)$ be that minimizer, which we refer to as a \emph{pseudo-target}.
Let $\X^*=\bigcup_{y\in\Y}x^*(y)$ be the set of all pseudo-targets for each $y\in\Y$.
Define $\Delta(\Y)$ to be the set of all mass distributions over $\Y$, and let $Q\in\Delta(\Y)$ be a distribution over targets.
We impose the following conditions on $Q$ and $\X^*$:

\begin{condition} \label{cond:ties-finite-x}
$Q(\{y: |x^*(y)| > 1\}) = 0$ and $|\X^*| < \infty$.
\end{condition}

The first part of Condition~\ref{cond:ties-finite-x}, consistent with \citet{nunez2017large}, allows us to ignore ties; relaxing this assumption would require a more intricate analysis.
The second part, requiring that the set of pseudo-targets is finite, is more fundamental to our approach and thus appears more challenging to relax; see \S~\ref{sec:discussion} for discussion.

\begin{definition}
    A \emph{position-optimization game} $\game$ is a game on $n$ agents satisfying Condition~\ref{cond:ties-finite-x}.
    The pure strategy of each agent $i \in [n]$ is a position $x_i\in\mathcal{X}$; let $\bm{x}=(x_1,\ldots,x_n)$ be the vector of $n$ players' positions. 
    Let $X_{\min}(\bm{x}, y)=\arg\min_{x_i \in \bm{x}}d(x_i,y)$ be the set of player positions that minimize proximity to target $y$.
    Let
    \begin{equation}\label{eq:pi}
        \pi_i(x_i, \bm{x}_{-i}, y) = \frac{\mathbbm{1}_{x_i\in X_{\min}(\bm{x}, y)}}{| X_{\min}(\bm{x}, y)|}
    \end{equation}
    represent the share of target $y$ that agent $i$ receives by playing location $x_i$, where target $y$ randomizes uniformly over its minimum proximity set $X_{\min}(\bm{x}, y)$.
    Then agent $i$'s utility in the game is
    \begin{equation}\label{eq:gen-utility}
        u_i(x_i,\bm{x}_{-i})=\E_{y \sim Q} \pi_i(x_i, \bm{x}_{-i}, y).
    \end{equation}
    That is, the objective of each of agent is to choose a location $x_i$ that minimizes the proximity function $d$ to $y\in\Y$ subject to the distribution $Q$ over targets and to the locations of the other players $\bm{x}_{-i}$.
\end{definition}

We can observe that $\game$ is a symmetric game, since for any permutation $\pi$ over the $[n]$ and for all $i \in [n]$, $u_{\pi(i)}(x_1, \ldots,x_i,\ldots, x_n) = u_i(x_{\pi(1)}, \ldots, x_{\pi(i)},\ldots, x_{\pi(n)})$. 

\paragraph{Mixed strategies.}
We will also consider mixed strategies, where each agent $i$ chooses a strategy $\sigma_i \in \Delta (\X)$.
A mixed strategy profile for all players is given by $\sigmaprofile = (\sigma_1, \dots, \sigma_n)$. 
Let $\sigmaprofile_{-i} = (\sigma_1, \dots, \sigma_{i-1}, \sigma_{i+1}, \dots, \sigma_n)$ indicate the strategy of all players other than $i$.
Under strategy profile $\sigma$, agent $i$ receives expected utility
\begin{equation}
    \label{eq:expected-utility-mixed}
    u_i(\sigma_i, \sigmaprofile_{-i}) = \E_{x_i \sim \sigma_i, x_{-i} \sim \sigmaprofile_{-i} } \left[u_i(x_i, \bm{x}_{-i}) \right].
\end{equation}

For any subset $I \subseteq \X$, $\sigma_i(I)$ denotes the probability that player $i$ chooses a reported position $x_i \in I$.
In an abuse of notation, for a single position $x \in \X$ we let $\sigma(x)$ denote $\sigma(\{x\})$.
We use the shorthand $u_i(x_i, \sigmaprofile_{-i}) = u_i(\ones[x_i], \sigmaprofile_{-i})$ to be the utility of agent $i$ when they play the pure strategy of always reporting position $x_i$.

We can now define several notions of Nash equilibria over the game $\game$.

\begin{definition}[Pure Nash equilibrium]
    A \emph{pure Nash equilibrium}  is a set of positions $\bm{x} \in \X^n$ such that for all $i \in [n]$ and all $x_i' \in \X$, $u_i(x_i',\bm{x}_{-i}) \leq u_i(x_i,\bm{x}_{-i})$.
\end{definition}

\begin{definition}[Mixed Nash equilibrium]
  A \emph{mixed Nash equilibrium} is a strategy profile $\sigmaprofile \in \Delta (\X)^n$ such that for all $i \in [n]$ and all $\sigma_i' \in \Delta (\X)^n$, $u_i(\sigma_i',\sigmaprofile_{-i}) \leq u_i(\sigma_i,\sigmaprofile_{-i})$.
\end{definition}

\begin{definition}[Symmetric Mixed Nash equilibrium]
  A \emph{symmetric mixed Nash equilibrium} is a single mixed strategy $\sigma \in \Delta (\X)$ such that $\sigmaprofile = (\sigma, \sigma, \ldots, \sigma)$ is a mixed Nash equilibrium. 
\end{definition}

We define strategy profiles as \emph{extreme} when agents play positions exclusively in $\X^*$:

\begin{definition}
    A pure strategy profile $\bm{x}$ is \emph{extreme} if for all $i \in [n]$, $x_i \in \X^*$.
    A mixed strategy profile $\sigmaprofile$ is extreme if for all $i \in [n]$, $\supp(\sigma_i) \subseteq \X^*$.
\end{definition}

We will later find in both pure and symmetric mixed settings that for large enough $n$, all equilibria are extreme.
It is thus useful for us to compare each equilibrium's distribution over $\X^*$ to $P$, the projection of $Q$ onto $\X^*$.
Formally, for an instance of $\game$, let
\[
P(x) = Q(\{y \in \Y: x^*(y) = x\})
\]
be the total mass of targets $y \in \Y$ for which the pseudo-target $x$ minimizes proximity.
In other words, $P$ projects the distribution $Q$ over targets onto the set of minimizers $\X^*$, so that we can relate the distribution of agent strategies over $\X^*$ back to $Q$.
We also denote $\pmin = \min_{x \in \X^*} P(x)$ as the pseudo-target with the smallest mass relative to $P$. 

\subsection{Motivating example models}
We can now outline how the model for $\game$ maps to the forecasting and Hotelling games introduced in \S~\ref{sec:intro}; for the others, we refer the reader to Appendix~\ref{appendix:setting}.
\begin{enumerate}
    \item \emph{Forecasting competitions.} 
    The $n$ strategic agents are forecasters.
    $\Y := \{0,1\}^m$, i.e. each target $y$ is a specific instantiation of the random event vector $Y$ with distribution $Q$; $\X := [0,1]^m$, so each position corresponds to a forecast over the random event vector $Y$; and $d(x, y) := \lVert x - y \rVert_2$.
    It follows that $\X^* = \Y$ with $x^*(y) = y$; and since $y = x^*(y)$, $P(y) = Q(y)$.
    That is, the set of pseudo-targets is the set of vertices (or event vectors) on the hypercube, and so the distribution over $\Y$ is thus exactly the distribution over $\X^* = \Y$.
    Since $|x^*(y)| = 1$ for all $y$ and $|\X^*| = |\Y|$, Condition~\ref{cond:ties-finite-x} is satisfied.
    The ex-ante utility of an agent is then 
    \[u_i(x_i, \bm{x}_{-i}) = \E_{y \sim Q} \frac{\ones_{d(x_i,y) \leq \max_{j \neq i} d( x_j,y)}}{\left| \arg\min_{j \in [n]} d(x_j, y) \right|}.\]

    \item \emph{Finite-location Hotelling games.}
    The $n$ strategic agents are retailers, playing in space $(\mathcal{S}, d).$
    $\X := \{1, \ldots, k\} \subset \mathcal{S}$ is a finite set of locations retailers can choose from.
    The target set $\Y$ is the (potentially non-finite) set of consumers, with corresponding distribution $\lambda$ over $\mathcal{S}$.
    The distance function is simply $d(x, y)$.
    Each consumer $y$ can be mapped to a (potentially non-strict) ranking $\pi_y$ over $\X$ induced by $d$. 
    The pseudo-target set for consumer $y$ is $x^*(y) = \{x \in \X: \pi_y(x) = 1\}$, so that $\X^* = \bigcup_{y \in \Y} \{x \in \X: \pi_y(x) = 1\}$ is finite. 
    Since the distribution $\lambda$ is assumed to be absolutely continuous, $\lambda(y: |x^*(y)| > 1) = 0$; therefore, Condition~\ref{cond:ties-finite-x} is satisfied.
    Then $P(x) = \lambda(\{ y \in \Y: \pi_y(1) = x\})$, i.e. $P(x)$ is the mass of consumers who prefer location $x$ first. 
    Let $\countv(x)$ be the number of agents in location profile $\bm{x}$.
    The utility of each retailer is
    \[u_i(x_i, \bm{x}_{-i}) = \frac{\lambda(\{ y \in \Y: \pi_y(x_i) \leq \pi_y(x_j) \; \forall j \neq i\})}{\countv(x_i)}.\]

    One could also map the more general setting of~\citet{nunez2016competing} to our game, where the sets $\X$ and $\Y$ need not be in the same space. 
    Here, each consumer is mapped explicitly to some ranking $\pi$ and the distance function is $d(x,y) = j$ s.t. $y(j) = x$. 
    Then we need to assume preferences are strict over the first position to guarantee that $x^*(y)$ is unique, since the underlying distribution of consumers is not assumed to be absolutely continuous. 

    \item \emph{Classic Hotelling games.}
    Again, the $n$ strategic agents are retailers, now playing in space $[0,1] =: \X$.
    $Q$ represents a discrete consumer distribution over $[0,1]$, and the target set $\supp(Q) =: \Y$ is the set of consumers.
    The distance function is $d(x, y) := |x - y|$, so that $x^*(y) = y$. 
    $|x^*(y)| = 1$ for all $y \in \Y$ and $\X^*$ is finite because $Q$ is a discrete distribution; thus, Condition~\ref{cond:ties-finite-x} is satisfied. 
    For each $x \in X^*$, $P(x) = Q(x)$, and 
    \[u_i(x_i, \bm{x}_{-i}) = \frac{Q(\{ y \in \Y: |x_i \leq x_j| \; \forall j \neq i\})}{\countv(x_i)}.\]
    
\end{enumerate}

\section{Pure Nash Equilibria} \label{sec:pure}
In this section, we study pure Nash equilibria, where agent strategies over positions are deterministic. 
Throughout, $\countv(x)$ denotes the number of agents on position $x \in \X$ for strategy profile $\bm{x} \in \X^n.$

\subsection{Initial observations}
We begin by making some observations about pure strategy profiles which will be useful later on to characterize equilibria.
First, our game is constant-sum, since
\begin{align*}
    \sum_{i \in [n]} u_i(x_i, \bm{x}_{-i}) &= \sum_{i \in [n]} \E_{y \sim Q}
    \frac{\mathbbm{1}_{x_i\in X_{\min}(\bm{x}, y)}}{| X_{\min}(\bm{x}, y)|}
    %&=  \E_{y \sim Q} \sum_{i \in [n]} \frac{\mathbbm{1}_{x_i\in X_{\min}(\bm{x}, y)}}{| X_{\min}(\bm{x}, y)|} \\
    = \E_{y \sim Q} \frac{1}{|X_{\min}(\bm{x}, y)|} \sum_{i \in [n]} \mathbbm{1}_{x_i\in X_{\min}(\bm{x}, y)}
    = 1.
\end{align*}
    
\begin{observation} \label{obs:utility-sum}
    For any strategy profile $\bm{x}$, $\sum_{i \in [n]} u_i(x_i, \bm{x}_{-i}) = 1$. 
\end{observation}

As a result, there exists an agent with utility at most $\frac{1}{n}$: otherwise, $u_i(x_i, \bm{x}_{-i}) > \frac{1}{n}$ for all agents $i$, and $\sum_{i \in [n]} u_i(x_i, \bm{x}_{-i}) > \sum_{i \in [n]} \frac{1}{n} = 1$.
Symmetrically, there exists an agent with utility at least $\frac{1}{n}$. 

\begin{observation} \label{obs:agent-low-util}
    Under any strategy profile $\bm{x}$, there exists an agent $i$ such that $u_i(x_i, \bm{x}_{-i}) \leq \frac{1}{n},$ and an agent $j$ such that $u_j(x_j, \bm{x}_{-j}) \geq \frac{1}{n}.$
\end{observation}

We can now relate the utility an agent gains by choosing a pseudo-target $x$ (whether as a strategy or deviation) to $P(x)$ and $\countv(x)$.
Intuitively, an agent playing some pseudo-target $x \in \X^*$ is always in the minimum set $X_{\min}$ for any of the targets $y$ associated with $x$, so that their utility share is at least $\frac{P(x)}{\countv(x)}.$
Using the same logic, if an agent then deviates to some $x' \in \X^*$, their new utility will be at least $\frac{P(x')}{\countv(x')+1}$.

\begin{definition}
    We say a strategy profile $\bm{x}$ \emph{covers $\X^*$} if $\countv(x) \geq 1$ for all $x \in \X^*$.
\end{definition}

\begin{lemma} \label{obs:pure-utility-bound}
    Consider a strategy profile $\bm{x}$ where for some agent $i$, $x_i \in \X^*.$
    Then
    $ u_i(x_i, \bm{x}_{-i}) \geq \frac{P(x_i)}{\countv(x_i)}$, and equality holds if $\bm{x}$ covers $\X^*$.
    For any deviation $x' \neq x_i$, $x' \in \X^*$, $ u_i(x', \bm{x}_{-i}) \geq \frac{P(x')}{\countv(x')+1}$, and equality holds if $\bm{x}$ covers $\X^*$ and $\countv(x_i) \geq 2$.
\end{lemma}

We also note that if a strategy profile covers $\X^*$, deviations to any position outside the set of pseudo-targets cannot be profitable.

\begin{lemma} \label{lem:covered-deviations}
    If $\bm{x}$ is a strategy profile that covers $\X^*$, then all strictly profitable deviations $x'$ satisfy $x' \in \X^*$ for any agent.
\end{lemma}

\begin{proof}
    Consider some deviation $x' \in \X \backslash \X^*$ from $x_i \in \X^*$, with $\Y_i = \{ y \in \Y: x^*(y) = x_i \}$ indicating the set of targets corresponding to pseudo-target $x_i$.
    Then since $\countv(x) \geq 1$ for all $x \in \X^*$ such that $x \neq x_i$,
    it follows that for all $y \in \Y \backslash \Y_i$, $\mathbbm{1}_{x'\in X_{\min}((x', \bm{x}_{-i}), y)} = 0$.

    We have two cases: $\countv(x_i) = 1,$ or $\countv(x_i) > 1$.
    If $\countv(x_i) > 1,$ then $\mathbbm{1}_{x'\in X_{\min}(\bm{x}, y)} = 0$ for all $y \in \Y_i$, and $u_i(x', \bm{x}_{-i}) = 0 \leq u_i(x_i, \bm{x}_{-i}).$
    If $\countv(x_i) = 1$, by Lemma~\ref{obs:pure-utility-bound} $u_i(x_i, \bm{x}_{-i}) = P(x_i)$. 
    Meanwhile, the most one can gain from deviation $x'$ is
    $u_i(x', \bm{x}_{-i}) \leq Q(\Y_i) = P(x_i) = u_i(x_i, \bm{x}_{-i})$.
\end{proof}

\subsection{Pure equilibrium characterization}

We begin our characterization by proving that all pure equilibria, where they exist, are extreme and cover pseudo-targets in the strict case where $n > \frac{1}{\pmin}$.
We show this bound is tight by giving an example of a non-extreme pure equilibrium for $n = \frac{1}{\pmin}$ in Proposition~\ref{prop:pure-non-extreme}, whose proof we leave to Appendix~\ref{appendix:pure-equilibria}.

\begin{theorem} \label{thm:pure-extreme}
For $\game$, if $n > \frac{1}{\pmin}$, then all pure equilibria are extreme and cover $\X^*$.
\end{theorem}
\begin{proof}
    We first show that for all $x \in \X^*$, $\countv(x) \geq 1$.
    Assume not, so that there is some $\hat x \in \X^*$ with $\countv(\hat x) = 0$.
    Then by Observation~\ref{obs:agent-low-util}, there exists some agent $i$ with utility $u_i(x_i, \bm{x}_{-i}) \leq \frac{1}{n}$.
    By Lemma~\ref{obs:pure-utility-bound}, agent $i$ can deviate to $\hat x$ and gain utility $u_i(\hat x, \bm{x}_{-i}) \geq \frac{P(\hat x)}{\countv(\hat x) + 1} \geq P(\hat x) > \frac{1}{n}$.
    
    Next we show any equilibrium profile is extreme.
    For contradiction, assume that there is some equilibrium with at least one agent $i$ playing $x_i \in \X \backslash \X^*.$
    Then since all targets are covered, for all $y \in \Y,$ $\mathbbm{1}_{x_i\in X_{\min}(\bm{x}, y)} = 0$.
    Thus $u_i(x_i, \mathbf{x}_{-i}) = 0$.
    For some $y \in \Y$, agent $i$ can then deviate to $x' = x^*(y)$ and gain win probability $u_i(x', \bm{x}_{-i}) \geq \frac{P(x')}{\countv(x')+1} > 0$.
\end{proof}

\begin{proposition} \label{prop:pure-non-extreme}
    There is a setting of $\game$ with $n = \frac{1}{\pmin}$ where an equilibrium exists that does not cover $\X^*$, and is not extreme.
\end{proposition}

Meanwhile, we prove existence of pure equilibria in settings where $n \geq \frac{2}{\pmin}$.
% By Lemma~\ref{lemma:no-pure-nash}, then, this bound on $n$ is tight.
Moreover, we show this bound is tight by giving examples of the game with no pure equilibria for $n$ below and arbitrarily close to $\frac{2}{\pmin}$ (Lemma~\ref{lemma:no-pure-nash-2/n}).
To prove existence, we construct an equilibrium explicitly using Algorithm \ref{alg:pure-construction}, which runs in time $O(|\X^*| n).$
Intuitively, the algorithm begins by assigning agents to each pseudo-target proportional to the underlying distribution $P$, such that there are at least two agents on each.
Then we perform a greedy routine, iteratively adding an agent to the pseudo-target that maximizes utility until we reach $n$ agents. 

\begin{algorithm} 
\caption{Pure equilibrium generator}
\label{alg:pure-construction}
    $n \geq \frac{2}{\pmin}$\;
    $n_0 \gets \sum_{x \in \X^*} \Bigl\lfloor \frac{2 P(x)}{\pmin} \Bigr\rfloor$\;
    $k_{n_0}(x) \gets \left\lfloor \frac{2 P(x)}{\pmin} \right\rfloor$ for all $x \in \X^*$\;
    \For{$t=n_0 + 1,\ldots,n$} {
    	$x_t^* \gets \text{argmax}_x \frac{P(x)}{k_{t-1}(x) + 1}$ \label{alg:argmax}\;
    	$x_t \gets x_t^*$\;
        $k_t(x_t^*) \gets k_{t-1}(x_t^*) + 1$\;
        $k_t(x) \gets k_{t-1}(x)$ for all $x \in \X^*$ s.t. $x \neq x_t^*$\;
    }
\end{algorithm}

\begin{theorem}\label{thm:pure-existence}
For $\game$, if $n \geq \frac{2}{\pmin}$, then a pure Nash equilibrium exists that is extreme with at least two agents on each pseudo-target.
%If $n > \frac{1}{\pmin}$, then all pure equilibria are extreme and cover $\X^*$;
If $n > \frac{2}{\pmin}$, then in all pure equilibria there are at least two agents on each pseudo-target.
\end{theorem}

% if n > 2/pmin, then a pure equilibrium is guaranteed to exist, and such equilibria are extreme.
% some kind of graphic for where equilibria are guaranteed to exist for the range n = 1/pmin to n > 2/pmin

\begin{proofsketch}
We give a brief sketch of correctness, and defer the full proof to Appendix \ref{appendix:pure-equilibria}.
First, note that for $n \geq \frac{2}{\pmin}$, $n \geq n_0$.
By construction, the initial assignment $\bm{x}^{(n_0)}$ is an extreme equilibrium.
Every position $x\in\X^*$ has at least two players assigned to it and deviating from their assigned position $x$ to some other $x'\in\X^*$ will strictly decrease a player's utility from $\frac{P(x_i)}{k_{n_0}(x_i)}$ to $\frac{P(x')}{\lfloor\frac{P(x')}{\pmin}\rfloor + 1}$ by Lemma~\ref{lem:covered-deviations}.
Then, note that if $\bm{x}^{(t)}$ is an extreme equilibrium $\bm{x}^{(t+1)}$ must be as well.
$\bm{x}^{(t+1)}$ is constructed by taking $\bm{x}^{(t)}$ and assigning the $(t+1)$st player to $x_t = \arg\max_x\frac{P(x)}{k_{t-1}(x)+1}$, the position that maximizes their utility.
They therefore have no incentive to deviate.
By symmetry, the other players assigned $x_t$ are also optimal, and any player on $x \neq x_t$ would not deviate since they were already playing optimally.
Therefore, by induction, every $\bm{x}^{(t)}$ for $t > n_0$ is also an extreme equilibrium, with at least two agents on each pseudo-target.
\end{proofsketch}

We now present counterexamples where no pure equilibria exist.
While~\citet{nunez2016competing} give an example of the finite-location Hotelling game with no pure equilibria where $n = 2$ and $\pmin = \frac{1}{3}$ (i.e. $n < \frac{1}{\pmin}$), we give examples of the game for any $\epsilon \in (0, \frac{1}{2}]$ with $\pmin = \frac{2-2\epsilon}{n}$ where no pure equilibrium exists, thus showing that our existence result in Theorem~\ref{thm:pure-existence} is tight.

\begin{lemma} \label{lemma:no-pure-nash-2/n}
    For any $\epsilon \in (0, \frac{1}{2}]$, there is an example of the game $\game$ with $\pmin = \frac{2-2\epsilon}{n}$ where no pure Nash equilibrium exists.
\end{lemma}

\begin{proofsketch}
    Consider a three-node graph with distribution $P$ on $\Y=\X=\X^*$, depicted for four separate strategy profiles, in Figure~\ref{fig:no-pure-equilibrium}.
    We are able to show that for any $\epsilon$ in the range, every equilibrium strategy profile $\bm{x}$ satisfies $1 \leq \countv(x_1),\countv(x_2) \leq 2$.
    Then we can demonstrate that for all such $\bm{x}$, some player has a profitable deviation.
\end{proofsketch}

\begin{figure}
\centering

% -------- LEFT COLUMN --------
\begin{minipage}[t]{0.48\textwidth}
\centering

\begin{tikzpicture}[
    scale=1.1,
    transform shape,
    node distance=2.2cm,
    every node/.style={
        draw,
        rectangle,
        minimum size=0.5cm,
        label distance=0pt
    }
]
    \node (x1) {$x_1$};

    \node (x2) [right of=x1] {$x_2$};

    \node (x3) [right of=x2] {$x_3$};

    \draw (x1) -- (x2) -- (x3);

    \fill ($(x1.north)+(0,6pt)$) circle (2pt);
    \fill ($(x2.north)+(0,6pt)$) circle (2pt);
    \fill ($(x3.north)+(0,6pt)$) circle (4pt);

    \draw[->, thick, bend left=30,
        shorten >=4pt, shorten <=4pt
    ] (x1.east) to (x2.west);

    % \draw[->, thick, bend right=30,
    %     shorten >=5pt, shorten <=5pt
    % ] (x3.north west) to node[midway, above, draw=none, font=\scriptsize] {(3)} (x1.north east);

    % \draw[->, thick, bend left=20,
    %     shorten >=4pt, shorten <=4pt
    % ] (x2.east) to node[midway, above, draw=none, font=\scriptsize] {(1)} (x3.west);
\end{tikzpicture}

\vspace{1cm}

\begin{tikzpicture}[
    scale=1.1,
    transform shape,
    node distance=2.2cm,
    every node/.style={
        draw,
        rectangle,
        minimum size=0.5cm,
        label distance=0pt
    }
]
    % second copy (identical)
    \node (x1) {$x_1$};
    \node (x2) [right of=x1] {$x_2$};
    \node (x3) [right of=x2] {$x_3$};

    \draw (x1) -- (x2) -- (x3);

    \fill ($(x1.north)+(0,6pt)$) circle (2pt);
    \fill ($(x2.north)+(-4pt,6pt)$) circle (2pt);
    \fill ($(x2.north)+( 4pt,6pt)$) circle (2pt);
    \fill ($(x3.north)+(0,6pt)$) circle (4pt);

    \draw[->, thick, bend left=30,
         shorten >=4pt, shorten <=4pt
     ] (x2.east) to (x3.west);

\end{tikzpicture}

\end{minipage}
\hfill
% -------- RIGHT COLUMN --------
\begin{minipage}[t]{0.48\textwidth}
\centering

\begin{tikzpicture}[
    scale=1.1,
    transform shape,
    node distance=2.2cm,
    every node/.style={
        draw,
        rectangle,
        minimum size=0.5cm,
        label distance=0pt
    }
]
    % third copy
    \node (x1) {$x_1$};
    \node (x2) [right of=x1] {$x_2$};
    \node (x3) [right of=x2] {$x_3$};

    \draw (x1) -- (x2) -- (x3);
    \fill ($(x2.north)+(0,6pt)$) circle (2pt);
    \fill ($(x1.north)+(-4pt,6pt)$) circle (2pt);
    \fill ($(x1.north)+( 4pt,6pt)$) circle (2pt);
    \fill ($(x3.north)+(0,6pt)$) circle (4pt);

        \draw[->, thick, bend left=30,
        shorten >=4pt, shorten <=4pt
    ] (x1.east) to (x2.west);
\end{tikzpicture}

\vspace{0.57cm}

\begin{tikzpicture}[
    scale=1.1,
    transform shape,
    node distance=2.2cm,
    every node/.style={
        draw,
        rectangle,
        minimum size=0.5cm,
        label distance=0pt
    }
]
    % fourth copy
    \node (x1) {$x_1$};
    \node (x2) [right of=x1] {$x_2$};
    \node (x3) [right of=x2] {$x_3$};

    \draw (x1) -- (x2) -- (x3);
    \fill ($(x1.north)+(-4pt,6pt)$) circle (2pt);
    \fill ($(x1.north)+(4pt,6pt)$) circle (2pt);
    \fill ($(x2.north)+(-4pt,6pt)$) circle (2pt);
    \fill ($(x2.north)+( 4pt,6pt)$) circle (2pt);
    \fill ($(x3.north)+(0,6pt)$) circle (4pt);

    \draw[->, thick, bend left=40,
        shorten >=4pt, shorten <=4pt
    ] (x1.east) to (x3.west);
\end{tikzpicture}

\end{minipage}

\caption{All best-response dynamics for each strategy profile $\bm{x}$ in the proof of Lemma~\ref{lemma:no-pure-nash-2/n}.
The dots represent the mass of players on each position $x_1,x_2,$ or $x_3$. Because in each $\bm{x}$ some player has a profitable deviation, no such $\bm{x}$ can be an equilibrium. In all examples, $P(x_1)=\frac{2-2\epsilon}{n},$ $P(x_2)=\frac{2-\epsilon}{n},$ and $P(x_3)=\frac{n-4+3\epsilon}{n}$.}
\label{fig:no-pure-equilibrium}
\end{figure}
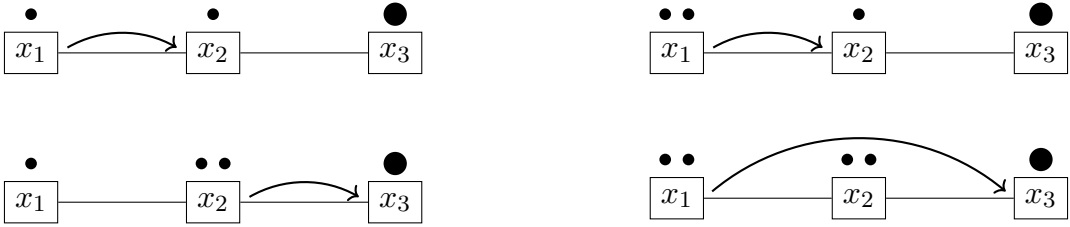

Interestingly, in the specific case of forecasting competitions with $m$ events, we are only able to give examples of non-existence for $\pmin \leq \frac{3}{2 n}$. 
Intuitively, our counterexamples construct a pseudo-target $x$ that must have a single agent $i$  playing $x$.
Agent $i$ can then deviate to another pseudo-target $x'$, while still sharing utility from the original position at $x$. 
In forecasting competitions, however, this deviation strategy breaks down for higher values of $\pmin$, as agent $i$ is required to share the utility of their original position $x$ with its $m$ closest pseudo-target neighbors in the hypercube. 
We conjecture that this bound on $\pmin$ for pure equilibrium existence in forecasting competitions is tight.

\begin{lemma} \label{lemma:no-pure-nash-forecasting}
    For any $\epsilon \in (\frac{1}{2}, 1)$, there is an example of a forecasting competition with $\pmin = \frac{2-\epsilon}{n}$ where no pure Nash equilibrium exists.
\end{lemma}

\begin{proof}
    Consider the forecasting competition with $m = 2$, so that $\Y = \{0,1\}^2$.
    WLOG, consider the resulting hypercube with counterclockwise vertices labeled $x_1 = \{1, 0\}$, $x_2 = \{0,0\}$, $x_3 = \{0, 1\}$, and $x_4 = \{1, 1\}$.
    Let $P(x_1) = P(x_2) = \frac{2 - \epsilon}{n}$, $P(x_3) = \frac{2 - \epsilon_2}{n}$, and $P(x_4) = \frac{n-6+2\epsilon+\epsilon_2}{n}$, where $\epsilon < 1$ is fixed and $0 < \epsilon_2 < \epsilon$, where we specify the value of $\epsilon_2$ later.
    
    It follows that $\pmin = \frac{2 - \epsilon_1}{n}$; and because $\epsilon_1 < 1,$ $n > \frac{1}{\pmin}$.
    Thus, by Theorem~\ref{thm:pure-extreme}, any pure equilibrium $\bm{x}$ is extreme with  $\countv(x_k) \geq 1$ for each $k$.
    Moreover, note that there is at most one agent on each of the vertices $x_1, x_2$, and $x_3$: assume not, so there are at least two agents on some $x_k$ for $k \in \{1,2,3\}$.
    Then one of those agents $i$ can deviate to $x_4$ to gain utility at least $u_i(x_4, \bm{x}_{-i}) \geq \frac{n - 6 + 2\epsilon_1 + \epsilon_2}{n (n-3)} > \frac{2 - \epsilon_2}{2n} \geq u_i(x_i, \bm{x}_{-i})$ (where the middle inequality holds for $n > 3 + \frac{2(3 - 2\epsilon_1 - \epsilon_2)}{\epsilon_2}$), contradicting that $\bm{x}$ is an equilibrium. 

    In combination, then, we know that a pure equilibrium $\bm{x}$ must satisfy $\countv(x_1) = 1$, $\countv(x_2) = 1$, $\countv(x_3) = 1$, and $\countv(x_4) = n-3$.
    It follows that agent $i$ on $x_2$ has utility $u_i(x_i, \bm{x}_{-i}) = \frac{2 - \epsilon}{n}$ -- but if they deviate to $x_3$, they receive utility $u_i(x_3, \bm{x}_{-i}) = \frac{2 - \epsilon}{3n} + \frac{2 - \epsilon_2}{2n}$, sharing the utility on $x_3$ with $\countv(x_3)$ other agents as well as the utility on $x_2$ with the agents on all of $x_2$'s neighbors, $x_1$ and $x_3$.
    For this deviation to $x_3$ to be profitable, then, we must have
    \begin{align*}
        u_i(x_3, \bm{x}_{-i}) = \frac{2 - \epsilon}{3n} + \frac{2 - \epsilon_2}{2n} &> \frac{2 - \epsilon}{n} = u_i(x_2, \bm{x}_{-i}) \\
        \frac{2}{3} \epsilon - \frac{\epsilon_2}{2} &> \frac{1}{3},
    \end{align*} 
    or (since $\epsilon_2$ can be taken arbitrarily close to $0$), $\epsilon > \frac{1}{2}.$
\end{proof}

\subsection{Equilibrium convergence and rates}
Consider the setting where $n > \frac{2}{\pmin}$.
By Theorems~\ref{thm:pure-extreme} and~\ref{thm:pure-existence}, there exists an equilibrium, and any equilibrium $\bm{x}$ must be extreme and cover all pseudo-targets.
Since in any equilibrium $\bm{x}$, all reported positions $x_i$ must correspond exactly to some pseudo-target $x \in \X^*$, we can define the \textbf{empirical} distribution over pseudo-targets induced by an equilibrium.

\begin{definition} Let the empirical distribution, $\emp(x)$, be the distribution defined over pseudo-targets $x \in \X^*$ induced by equilibrium profile $\bm{x}$. 
That is, $\emp(x)=\frac{\countv(x)}{n}$.
\end{definition}

In any equilibrium profile $\bm{x}$, $\emp(x)$ is a well-defined distribution: for all $x \in \X^*$, $\frac{\countv(x)}{n}\in[0,1]$. 
Furthermore, Theorem~\ref{thm:pure-existence} implies that $\sum_{x\in \X^*}\countv(x)=n$, and so $\sum_{x \in \X^*} \emp(x)=1$.

Now we have the tools to prove our main wisdom result.
We first show that in equilibrium, all pseudo-targets are covered by a number of agents as a function of a more general lower bound on $n$ (Lemma~\ref{lem:pure-lower-bound}).
This allows us to tightly bound the utility of any player in equilibrium in Lemma~\ref{lem:wp-convergence} as a function of $n$, which in turn translates to a convergence rate on $\emp$ to $P$.
We leave proofs of the Lemmas to Appendix~\ref{appendix:pure-equilibria}.

\begin{lemma}\label{lem:pure-lower-bound}
Assume $n \geq \frac{1}{c} > \frac{2}{\pmin}$.
Then in any equilibrium $\bm{x}$, for all $x \in \X^*$, the number of agents on $x$ satisfies $\countv(x)\geq\lfloor cn\rfloor$.
\end{lemma}

\begin{lemma}\label{lem:wp-convergence}
Assume $n \geq \frac{1}{c} > \frac{2}{\pmin}$.
Then, in an equilibrium $\bm{x}$, the win probability of any player must satisfy $u_i(x_i, \mathbf{x}_{-i}) \in \left[\frac{1}{n}\frac{\lfloor cn\rfloor}{\lfloor cn\rfloor +1},\frac{1}{n}\frac{\lfloor cn\rfloor+1}{\lfloor cn\rfloor}\right]$.
\end{lemma}

\begin{theorem}\label{thm:kl-div}
Assume $n \geq \frac{1}{c} > \frac{2}{\pmin}$.
Consider any equilibrium $\bm{x}$ of the game $\game$, and let $\kldiv{P}{Q}$ denote the $KL$ divergence of distribution $P$ from $Q$.
Then
\[\kldiv{\emp}{P} \leq \log \left( \frac{\lfloor cn\rfloor + 1}{\lfloor cn \rfloor} \right).\]
\end{theorem}

\begin{proof}
Consider the empirical distribution $\emp$ induced by equilibrium $\bm{x}$ of the game $\game$.
Then
\[
\kldiv{\emp}{P} = \sum_{x \in \X^*}\frac{\countv(x)}{n}\log \left( \frac{\countv(x)}{nP(x)} \right) \leq \sum_{x \in \X^*}\frac{\countv(x)}{n}\left| \log \left( \frac{\countv(x)}{nP(x)} \right) \right|.
\]
By Lemma~\ref{lem:wp-convergence}, for all $x \in \X^*$ we have that
\begin{equation} \label{eq:log-upper-bound}
\frac{\lfloor cn\rfloor}{\lfloor cn \rfloor + 1} \leq \frac{\countv(x)}{{n P(x)}} \leq \frac{\lfloor cn\rfloor + 1}{\lfloor cn \rfloor}.
\end{equation}
Combining Inequality~\ref{eq:log-upper-bound} with the fact that $|\log(z)| = \log(1/z)$ for $0 < z < 1$, we then have $\left| \log \left( \frac{\countv(x)}{nP(x)}\right) \right| \leq \frac{\lfloor cn\rfloor + 1}{\lfloor cn \rfloor}$ for all $x \in \X^*$.
Thus 
\begin{align*}
    \kldiv{\emp}{P} &\leq \sum_{x \in \X^*} \frac{\countv(x)}{n} \log \left( \frac{\lfloor cn\rfloor + 1}{\lfloor cn \rfloor} \right) \\
    &= \log \left( \frac{\lfloor cn\rfloor + 1}{\lfloor cn \rfloor} \right) \left( \frac{1}{n} \sum_{x \in \X^*} \countv(x) \right) \\
    &= \log \left( \frac{\lfloor cn\rfloor + 1}{\lfloor cn \rfloor} \right).
\end{align*}
\end{proof}

\section{Symmetric Mixed Equilibria}\label{sec:mixed-equilibria}
In this section, we study symmetric mixed equilibria, where $\sigma_i=\sigma$ for all players $i$.
In a slight abuse of notation, when discussing symmetric mixed strategies we will use $\sigma_{-i} \in \Delta(\X)^{n-1}$ to denote the strategy profile of all agents but $i$ playing $\sigma,$ i.e. $\sigma_{-i} = (\sigma, \sigma, \ldots, \sigma).$

We will derive results analagous to \S~\ref{sec:pure} for symmetric mixed equilibria.
Specifically, we show for some lower bound on the number of agents $n$, we have (1) all equilibria are extreme and (2) an upper bound on the convergence rate of $\sigma$ to $P$.

\subsection{Existence and initial observations}
While the payoff function in our game is discontinuous due to ties, its nice symmetric structure allows us to explicitly state equilibrium payoffs.
Consider a symmetric game that is constant-sum.
By symmetry of the game, the utility of a symmetric mixed strategy profile must be equal for all agents; and since these utilities sum to $c$, the utility for each agent is $\frac{c}{n}$.
Since our game is symmetric and Observation~\ref{obs:utility-sum} shows it is constant-sum for $c = 1,$ we have that:

\begin{observation}\label{obs:symmetric-equil-utility}
    For any symmetric strategy profile $\sigma$ in the game $\game$, we have that $u_i(\sigma,\bm{\sigma}_{-i})=\frac{1}{n}$ for all $i\in[n]$.
\end{observation}

Observation \ref{obs:symmetric-equil-utility} allows us to conclude existence of symmetric mixed equilibria by using a result from~\citet{reny1999existence} on discontinuous games. 
We leave proofs to Appendix~\ref{appendix:mixed-equilibria}.

\begin{theorem} \label{thm:mixed-existence}
For any $n > 1$, if $\X$ is finite or a compact Hausdorff space then the game $\game$ has a symmetric mixed equilibrium. 
\end{theorem}

\subsection{Equilibrium convergence and rates}\label{subs:events}

For ease of exposition in this section, we set $p_x := P(\{x\})$ and $\sigma_x := \sigma(\{x\})$.
We first present our quantitative convergence result for symmetric mixed equilibria.

\begin{theorem}\label{thm:mixed-convergence}
    For $n > \max \left\{ 43, 8 \left(\frac{4}{\pmin} \log \frac{1}{\pmin} \right) \right\},$ all symmetric mixed Nash equilibrium are extreme and satisfy $\lvert \sigma_x - p_x \rvert \leq \frac{1}{n}$ for all $x \in \X^*$. 
\end{theorem}

We proceed with a proof sketch of Theorem~\ref{thm:mixed-convergence}; we defer formal statements of supporting lemmas and their proofs to Appendix \ref{appendix:mixed-equilibria}.
Our high-level approach in characterizing the relationship between $\sigma_x$ and $p_x$ is to understand the utility an agent $i$ gains by playing a pure strategy $x_i=x$.
Intuitively, if $\sigma_x$ is too small relative to $p_x$, agent $i$ gains higher utility by deviating to $x$. 
If $\sigma_x$ is too large relative to $p_x$, agent $i$ loses too much utility deviating to $x$, contradicting that $x$ is in the support of the mixed strategy. 
We are able to characterize when utilities are too high or too low using our previous observation (Observation~\ref{obs:symmetric-equil-utility}) that in any equilibrium $\sigma$, $u_i(\sigma_i,\bm{\sigma}_{-i})=\frac{1}{n}$.

\paragraph{Breaking down the utility.}
We can break down the utility of playing pure strategy $x_i=x$ into three disjoint and exhaustive events $E_1,E_2$ and $E_3$, stated as conditions on the realizations of the target $y\in\Y$ and other agents' positions $x_{-i} \in \X^{n-1}$:

\begin{enumerate}
    \item $E_1$: $x^*(y)=x$, and $x_j \neq x$ for some $j\neq i$.
    (That is, $x$ is the pseudo-target corresponding to $y$, and at most $k \leq n-2$ others play $x$.)
    In this case, we derive the utility conditional on event $E_1$ by taking a binomial sum with probability $\sigma_x$ over each value $k$.
    \item $E_2$: $x_j=x$ for all $j$.
    In this case, the realization of $y$ is irrelevant, since everyone is playing the same position, and the expected utility from $x$ is exactly $\frac{1}{n}$.
    \item $E_3$: $x^*(y) \neq x$, and $x_j\neq x$ for some $j\neq i$.
    In this case, agent $i$ can still gain positive utility so long as $x \in X_{\min}(\bm{x}, y)$, i.e., $d(x_j,y)\geq d(x,y)$ for all $j\neq i$.
\end{enumerate}

\paragraph{Obtaining a lower bound on $\sigma_x$.}
We find that the expected utilities on $E_1$ and $E_2$ can be computed explicitly. In Lemma \ref{lem:deviation-lower-bound}, we observe that the sum of these terms then serve as a lower bound on the utility of playing position $x$, with
\begin{align}
    u_i(x, \bm{\sigma}_{-i})
    &\geq \E[\pi_i(x_i, \bm{x}_{-i}, y) \ones_{E_1}] + \E[\pi_i(x_i, \bm{x}_{-i}, y) \ones_{E_2}]
    \\
    \label{eq:u-inequality}
    &= \frac{p_x}{n\sigma_x}(1-\sigma_x^n-(1-\sigma_x)^n)+\frac{1}{n}\sigma_x^{n-1}.
\end{align}
Let $g(p_x, \sigma_x) = \frac{p_x}{n\sigma_x}(1-\sigma_x^n-(1-\sigma_x)^n)+\frac{1}{n}\sigma_x^{n-1}.$

Since $x$ is in the support of $\sigma$, $u_i(x,\bm{\sigma}_{-i}) = \frac{1}{n}$ by Observation~\ref{obs:symmetric-equil-utility}; in combination with Inequality~\ref{eq:u-inequality}, we must then have $g(p_x, \sigma_x) \leq \frac{1}{n}.$
% \[p_x=\frac{\sigma_x-\sigma_x^n}{1-\sigma_x-(1-\sigma_x)^n}.\]
We find that this inequality corresponds exactly to the inverse relation $\sigma_x \geq G^{-1}(p_x)$, where
\begin{equation}\label{eq:G}
    G(\sigma_x)=\frac{\sigma_x-\sigma_x^n}{1-\sigma_x-(1-\sigma_x)^n}.
\end{equation}
$G^{-1}(p_x)$ is well-defined for $p_x \in \left( \frac{1}{n}, 1 - \frac{1}{n}\right)$.
We can bound the distance between $p_x$ and $G^{-1}(p_x)$ by $\frac{1}{n}$ (Lemma~\ref{lem:G-equality}).
Chained together, then, $\sigma_x \geq G^{-1}(p_x) \geq p_x - \frac{1}{n}$ for $p_x \in(\frac{1}{n},1-\frac{1}{n})$.
Meanwhile, for pseudo-targets with higher mass, i.e. $p_x\geq1-\frac{1}{n}$, Lemma \ref{lem:high-p-equilibrium-2} proves that the only symmetric equilibrium is $\sigma_x=1$; and if $p_x \leq \frac{1}{n}$, the inequality holds trivially.
Combining each of these cases gives us Proposition \ref{prop:mixed-lower-bound}: for all $p_x \in [0, 1],$ $\sigma_x\geq p_x-\frac{1}{n}$.

\paragraph{Obtaining an upper bound on $\sigma_x$.}
From Equation~\ref{eq:u-inequality}, we have exact expressions for $\E[\pi\ones_{E_1}] + \E[\pi\ones_{E_2}]$.
To give an upper bound on $\sigma_x$, therefore, it suffices to upper bound $\E[\pi\ones_{E_3}]$.
That is, we aim to upper bound the utility share when agent $i$ is not on the corresponding pseudo-target for $y$, and at least one other agent is at a different location than agent $i$.
A useful related event to consider is $C := \{\X^*$ is covered$\} = \{\cup_i \{x_i\} = \X^*\}$.
In particular, if \emph{every} other pseudo-target $x' \in \X^*$ is covered by some agent, someone else is always closer in proximity to $y$ and $i$ gains no utility.
Meanwhile, if $\X^*$ is not covered, we crudely upper bound $i$'s utility by assuming they \emph{always} win $y$.
We thus have, 
\begin{align}
\label{eq:upper-bound-C}
    \E[\pi_i(x_i, \bm{x}_{-i}, y) \ones_{E_3}] \leq \E[\ones_{E_3 \setminus C}] \leq 1 - \Pr[C].
\end{align}

Lemma~\ref{lem:mixed-extreme} shows $1 - \Pr[C] \leq \frac{1}{n^2}$ for some lower bound on $n$ which is a function of $\pmin$.
The proof follows standard coupon-collector arguments, using the lower bound on $\sigma_x$ above.
The same logic shows the extremeness of $\sigma$ (Proposition \ref{prop:extreme-equilibria}), where now we consider $x\notin \X^*$, and note that $x$ has utility 0 on $C$.
Thus $\E[\pi_i(x, \bm{x}_{-i}, y)] \leq 1 - \Pr[C] \leq \tfrac 1 {n^2} < \tfrac 1 n$, meaning $x\notin\supp(\sigma)$ by Observation~\ref{obs:symmetric-equil-utility}.

Returning to the upper bound on $\sigma_x$, we now combine eq.~\eqref{eq:upper-bound-C} with eq.~\eqref{eq:u-inequality}, giving

\begin{align}
\label{eq:upper-bound-final-expression}
    u_i(x,\bm{\sigma}_{-i}) \leq \frac{p_x}{n\sigma_x}\left(1-\sigma_x^n-(1-\sigma_x)^n\right)+\frac{1}{n}\sigma_x^{n-1}+\frac{1}{n^2}.
\end{align}    
We can thus proceed using a similar argument as the lower bound.
Specifically, in equilibrium, we know $u_i(x, \bm{\sigma}_{-i}) = \frac{1}{n},$ allowing us to lower bound the right-hand side of eq.~\eqref{eq:upper-bound-final-expression} by $\frac{1}{n}$ to recover a relationship between $\sigma_x$ and $p_x$ (Lemma~\ref{lem:lower-bound-util}).
Unfortunately, the resulting inverse relation is not as nice: monotonicity breaks when $\sigma_x$ is too large, intuitively because agents can always achieve utility $\frac{1}{n}$ by all playing $x$; with the addition of utility conditioned on $E_3$, a player can secure utility above $\frac{1}{n}$.
In light of this technical challenge, we are able to prove the upper bound in Proposition~\ref{prop:mixed-upper-bound} through careful casework.

\paragraph{An example when $|\X^*| = 2$.}
As an illustrative example, we depict our analysis in the fundamental forecasting setting where forecasters submit predictions for a single binary-outcome event $y$, i.e., when $m=1$ (and hence $|\X^*|=2$); concretely, forecasters' symmetric equilibrium strategy $\sigma_x$ is the weight they assign to the pure strategy of reporting $x=1$.
We consider the case where $n$ satisfies the bound in Proposition \ref{prop:extreme-equilibria}, so that equilibria are extreme.
In this setting, we are able to pin down the symmetric mixed equilibrium exactly.
%Let $\X^*=\{x,\overline{x}\}$.
In particular, the utility is given by eq.~\eqref{eq:u-inequality} as an equality, since the expected utility conditioned on $E_3$ is 0; i.e.,\[u_i(x,\bm{\sigma}_{-i})=\frac{p_x}{n\sigma_x}(1-\sigma_x^n-(1-\sigma_x)^n)+\frac{1}{n}\sigma_x^{n-1}.\]
When the expression for $u_i(x,\bm{\sigma}_{-i})$ is exact, the strategy $\sigma_x$ satisfies $\sigma_x = G^{-1}(p_x)$. 
Meanwhile, for $p_x<\frac{1}{n}$ and $p_x>1-\frac{1}{n}$, the only equilibria are to play $\sigma_x=0$ and $\sigma_x=1$, respectively, by Lemma \ref{lem:high-p-equilibrium-2} (which automatically applies when $p_x<\frac{1}{n}$ as $1-p_x>1-\frac{1}{n}$).

Figure~\ref{fig:fig1} explicitly illustrates the extremizing phenomenon that is observed in real-life forecasting competitions: for any $p_x\neq0.5$, $\sigma_x$ places \emph{more} weight than $p_x$ on the outcome of $x$ that is more likely under $P$. 
%as $p_x$ moves closer to e.g. 0, the equilibrium strategy $\sigma_x$ is closer still (and likewise for $y=1$).
%Perhaps surprisingly, this example is evidence that even in a highly stylized full-information, single-event forecasting competition, we can recover .
The striking emergence of non-truthful behavior even under our stylized position-game model suggests that this example encapsulates the core insight of strategic behavior in more complex settings.

\begin{figure}[h]
    \centering
    \includegraphics[scale=0.5]{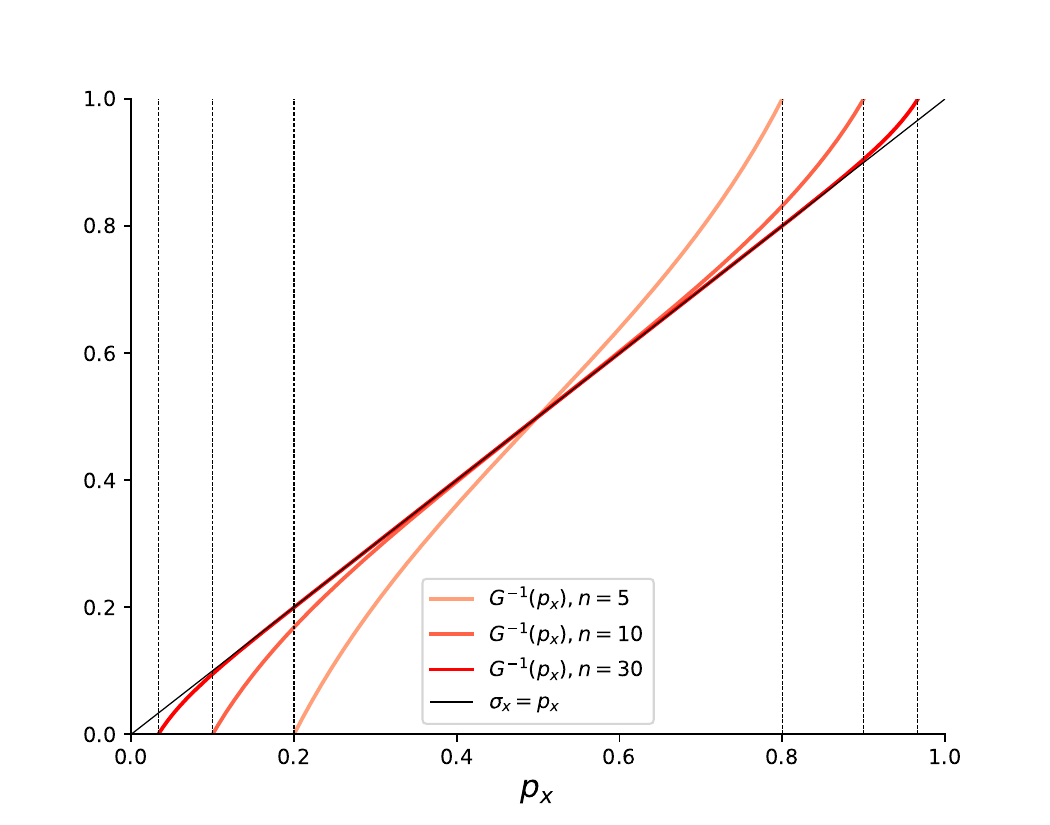}
    \caption{Plot of $\sigma_x$ as a function of $p_x$ for a symmetric strategy $\sigma$ and pseudo-target $x$, for the case $|\X^*| = 2$, for various values of $n$.
    Red lines indicate the function $\sigma_x = G^{-1}(p_x)$, and the black line indicates where $\sigma_x= p_x$; vertical lines indicate $\sigma_x=\frac{1}{n}$ and $\sigma_x=1-\frac{1}{n}$ for the different values of $n$, respectively. 
    Observe that for increasing values of $n$, the equilibrium $\sigma_x$ given by $G^{-1}(p_x)$ hews closer and closer to $p_x$.}
    \label{fig:fig1}
\end{figure}

\section{Applications} \label{sec:app}
In this section, we aim to illustrate how the results for the general abstract game $\game$ apply to each of the specific games introduced in \S~\ref{sec:intro}, and what the results imply in each of their different contexts.

\subsection{Forecasting competitions}
We begin by outlining results for forecasting competitions.
As a reminder, forecasters submit predictions $x_i \in [0,1]^m$ over a random vector $Y \in \{0,1\}^m$ of $m$ binary events with distribution $Q$, which is common-knowledge.
$\Y := \{0,1\}^m$ is the set of event vector realizations, with $\X^* = \Y$, so that a strategy profile is extreme if forecasters only submit deterministic marginal predictions of 0 or 1 for each event. 
Let $P$ be the projection of $Q$ onto $\X^*$.
Note that $\pmin \leq 2^{-m},$ so that $n \geq \frac{2}{\pmin}$ implies $n \geq 2^{m+1}$.

\begin{corollary}{(Forecasting)}
    If $n > \frac{1}{\pmin}$, then every pure equilibrium $\bm{x}$ is extreme (for each forecaster $i$, $\bm{x}_i \in \{0,1\}^m$) and covers $\Y$ (for each $y \in \{0,1\}^m$, there exists a forecaster $i$ such that $\bm{x}_i = y$). 
    If $n \geq \frac{2}{\pmin}$, then there exists an extreme, pure Nash equilibrium such that for all $y\in \{0,1\}^m$, $k_{\bm{x}}(y)\geq2$, and if $n>\frac{2}{\pmin}$, then this condition holds for all pure Nash equilibria.
    Moreover, if $n > \frac{1}{c} \geq \frac{2}{\pmin}$, for any pure equilibrium $\bm{x}$,
    \[\kldiv{\emp}{P} \leq \log \left( \frac{\lfloor cn\rfloor + 1}{\lfloor cn \rfloor} \right),\]
    where $\hat{P}_{\bm{x}}$ is the empirical distribution on $\{0,1\}^m$ induced by $\bm{x}$.

    Symmetric mixed Nash equilibria exist for any $n > 1$.
    If $n > \max \left\{ 43, 8 \left(\frac{4}{\pmin} \log \frac{1}{\pmin} \right) \right\}$, any symmetric mixed Nash equilibrium $\sigma$ is extreme, and for all $y \in \{0,1\}^m$, we have $|\sigma(y) - Q(y)| \leq \frac{1}{n}$.
\end{corollary}

In other words, when the number of forecasters is large relative to the number of events, forecasters will not submit their true predictions about $Y$ under their belief distribution $Q$. 
Instead, forecasters concentrate their strategies on the potential \emph{outcomes} of the random events.
When agents employ pure strategies, the resulting equilibrium profile will then mirror the underlying distribution $Q$ over $Y$, so that while individually forecasters do not submit reports drawn from $Q$, collectively we have a wisdom-of-the-crowds effect: the empirical distribution of reports will closely resemble $Q$.

Since many forecasting competitions are one-shot with little context to predict the behavior of other players, symmetric mixed equilibria are reasonable in modeling participant behavior.
In light of our results, we expect in large $n$ regimes that forecasters will sample an outcome from the underlying event distribution, and submit a fully extreme predictions that assign probability 1 to that outcome.
Interestingly, the winner of the competition is not someone who submits a well-calibrated forecast, but instead gets lucky and beats out others by sampling the correct outcome.
From an outsider's point of view, then, there is little insight to gain about the underlying event probabilities from the winner's (or any individual's) reports.
Instead, one can only empirically reconstruct the underlying distribution by observing enough reports from different forecasters that they can be treated as samples from $Q$.
Conversely, our results imply that extreme forecasts are not necessarily the work of lazy or uninformed participants making wild guesses; rather, they may reflect sophisticated strategies designed to maximize win probability by sampling from a highly accurate distribution.

\subsection{Finite-location Hotelling games}

Now consider the finite-location Hotelling game, with retailers choosing locations in $\{1, 2, \ldots, k\} \subset \mathcal{S}$ and the set of consumers $\Y$ following a distribution $Q$ over $\mathcal{S}$. 
Here the pseudo-target set $\X^*$ corresponds to the set of locations that are ranked first by at least one consumer (i.e., $x^*(y) = \{ x \in \X: \pi_y(x) = 1$), and $P$ is the projection of $Q$ onto $\X^*$.
Our results translate as follows:

\begin{corollary}{(Finite-location Hotelling)}
    If $n > \frac{1}{\pmin}$, then every pure equilibrium is extreme (for each retailer $i$, there exists some consumer $y \in \Y$ such that $\pi_y(\bm{x}_i) = 1$) and covers the set of first-ranked locations $\X^*$ (for each $y \in \Y$, there exists a retailer $i$ such that $\pi_y(\bm{x}_i) = 1$).
    If $n \geq \frac{2}{\pmin}$, then there exists an extreme, pure Nash equilibrium such that for all $x\in\X^*$, $k_{\bm{x}}(x)\geq2$, and if $n > \frac{2}{\pmin}, $ then this condition holds for all pure Nash equilibria.
    Moreover, for $n > \frac{1}{c} \geq \frac{2}{\pmin}$,
    \[\kldiv{\emp}{P} \leq \log \left( \frac{\lfloor cn\rfloor + 1}{\lfloor cn \rfloor} \right),\]
    where $\hat{P}_{\bm{x}}$ is the empirical distribution of retailer locations induced by an equilibrium strategy profile $\bm{x}$.

    Symmetric mixed Nash equilibria exist for any $n > 1$.
    If $n > \max \left\{ 43, 8 \left(\frac{4}{\pmin} \log \frac{1}{\pmin} \right) \right\}$, any symmetric mixed Nash equilibrium is extreme, and for each first-ranked position $x \in \X^*$, we have $|\sigma(x) - P(x)| \leq \frac{1}{n}.$
\end{corollary}

That is, when the number of retailers is large relative to the underlying distribution over consumer preferences, they will concentrate on the most-preferred set of locations. 
With sufficient competition, then, retailers hyper-specialize over the consumers' preferences. 
Our conclusions mirror the results shown in~\citet{nunez2016competing} and~\citet{nunez2017large} for the finite-location Hotelling game, though the convergence rates are novel. 

\subsection{Classic Hotelling games}
In the classic Hotelling game, retailers choose locations in $\X := [0,1]$ and consumers follow a distribution $Q$ with finite support over $[0,1]$.
Then $\X^* = \supp(Q)$ corresponds to the set of consumer locations. 
Our result then implies the following.

\begin{corollary}{(Classic Hotelling)}
    If $n > \frac{1}{\pmin}$, then every pure equilibrium is extreme (for each retailer $i$, $\bm{x}_i \in \supp(Q)$) and covers the consumers (for each $y \in \supp(Q)$, there exists a retailer $i$ such that $\bm{x}_i = y$).
    If $n \geq \frac{2}{\pmin}$, then there exists an extreme, pure Nash equilibrium such that for all $y\in\supp(Q)$, $k_{\bm{x}}(y)\geq2$, and if $n > \frac{2}{\pmin}, $ this condition holds for all pure Nash equilibria.
    Moreover, for $n > c \geq \frac{2}{\pmin}$,
    \[\kldiv{\emp}{P} \leq \log \left( \frac{\lfloor cn\rfloor + 1}{\lfloor cn \rfloor} \right),\]
    where $\hat{P}_{\bm{x}}$ is the empirical distribution of retailer locations induced by an equilibrium strategy profile $\bm{x}$.

    Symmetric mixed Nash equilibria exist for any $n > 1$.
    If $n > \max \left\{ 43, 8 \left(\frac{4}{\pmin} \log \frac{1}{\pmin} \right) \right\}$, any symmetric mixed Nash equilibrium $\sigma$ covers the consumers, and for each consumer location $y \in \supp(Q)$, we have $|\sigma(y) - P(y)| \leq \frac{1}{n}.$
\end{corollary}

Our theorem, when applied to the classic Hotelling game, provides a discrete analogue to the results of \citet{osborne1986nature}. Whereas they prove convergence to the consumer distribution under the assumption that its CDF is twice differentiable, we prove convergence in the case where the distribution is discrete.

\section{Discussion} \label{sec:discussion}
In this paper, we present the position-optimization game: players select positions in a common space, with the goal of capturing a set of targets according to some proximity function.
We fully characterize pure and symmetric mixed equilibria for a large number of players. 
Furthermore, we show that equilibria converge in $n$ to the projection of the target distribution onto the position space, and characterize convergence rates. 
These results have broad implications for games encompassed by our model.
In particular, we are the first to provide convergence rates; the first to fully characterize equilibria in traditional forecasting competitions; and the first to derive large-$n$ results for the classic Hotelling game with a discrete consumer distribution.

\paragraph{Convergence in all settings.}
A natural question is whether the requirement $|\X^*|<\infty$ can be relaxed, in particular to the case where $\X^*$ is a continuum like $[0,1]$.
Our current proof techniques first show that (pure or mixed) equilibrium strategies are extreme and cover $\X^*$; we then leverage this condition to analyze the ratio of the $Q$-mass on any $x^*(y)\in \X^*$ to the (pure or mixed) strategy mass on $x^*(y)$.
Using this ratio of masses, we control the utilities and show convergence rates.
When $\X^*$ is a continuum, however, this ratio is not generally well-defined.

One approach to a coverage-style argument in the continuous setting is to consider a finite partition of the position space, but this runs into "leakage" problems where an agent's utility depends nontrivially not only on the placement of other agents in that bin, but on the placement of agents in other bins.
\citet{osborne1986nature} prove that symmetric mixed equilibria converge to $Q$ for continuous (bounded) $\X^*$, but these results lack convergence rates and furthermore place three key assumptions:
that $Q$ is atomless and differentiable; the equilibrium is differentiable; and its limit over $n$ exists.
Characterizing convergence rates without these conditions presents a significant technical challenge, as bounding agents' utilities requires carefully controlling various Taylor approximation errors while also managing the aforementioned leakage problems.
The difficulty of relaxing $|\X^*|<\infty$ notwithstanding, given the consistency of our results with prior work, we conjecture that this convergence will hold for all position-optimization games in much greater generality, even with a significant relaxation of Condition~\ref{cond:ties-finite-x}.

\paragraph{Small $n$ regimes.}
We think of this paper as characterizing equilibria for \emph{large} values of $n$, where the number of players overwhelms any underlying structure of the game.
A natural next step is to analyze equilibrium behavior for smaller values of $n$.
For \emph{medium} values of $n$, we conjecture that the structure of the game becomes much more important.
For example, in forecasting competitions, the geometry of the metric space $[0,1]^m$ coupled with the $\ell_2$ distance means that forecasters may be able to capture more win probability by playing in the interior of the hypercube.
Can we characterize at what point over $n$ forecasters depart from the vertices, and whether they would remain close to them or move quickly toward less extreme points?

\bibliographystyle{ACM-Reference-Format}
\bibliography{references}

\newpage
\appendix

\section{Additional applications omitted from \S~\ref{sec:setting}}\label{appendix:setting}

Here, we present additional applications of our model.
\begin{enumerate}
    \item \emph{Spatial voting.}
    The $n$ strategic agents are candidates.
    $\Y$ is a finite set of $m$ voters.
    $\X$ is a set of potential candidates' ideological positions; for example, $\X := [0, 1]^2$ could represent an ideological spectrum over two different issues.
    Each candidate chooses a position $x \in \X$.
    Then $d: \X \times \Y \to \reals_{\geq 0}$ corresponds to a function that maps each position and voter to some cost.
    We assume there is a single most-preferred position $x$ for each voter $y$ such that $d(x,y) = 0$, with $d(x,y) < d(x',y)$ for all $x' \neq x$. 
    The pseudo-target for voter $y$ is the position $x$ such that $d(x,y) = 0$, so that $\X^* = \bigcup_{y \in \Y} \arg\min_{x \in \X} d(x, y)$.
    $\X^*$ is finite because $\Y$ is finite, so Condition~\ref{cond:ties-finite-x} is satisfied.
    
    The winning candidate is selected according to plurality rule: the candidate who is ranked first relative to $d$ for the most voters wins (with randomized tie-breaking).
    $Q(y) := \frac{1}{|\Y|}$ is thus the amount of influence of voter $y$ from each candidate's perspective.
    Then 
    \[P(x) = \frac{\sum_{y \in \Y} \ones_{x = \arg\min_{x \in \X} d(x, y)}}{|\Y|},\]
    i.e. $P(x)$ is the percentage of voters who most prefer position $x$.
    It follows that the utility of each candidate $i$ is
    \[u_i(x_i, \bm{x}_{-i}) = \frac{1}{|\Y| \countv(x_i)} \sum_{y \in \Y} \ones_{d(x_i, y) \leq d(x_j, y) \; \forall j \neq i}.\]

    % If we assumed $\X$ was finite; every position was preferred most by at least one voter; and ties were not allowed, the spatial voting model would be a specific instantiation of the model in~\citet{nunez2016competing}. 
    
    \item \emph{Discrete Voronoi games.} 
    Let $G = (V, E)$ be an undirected, vertex-weighted graph.
    There are $n$ strategic influencers who want to attract users on the vertices, with position set $\X := V$.
    Targets correspond to users, with $\Y := V$, and the proximity function is $d(x, y)$ for $d$ the shortest path between $x$ and $y$.
    The distribution over targets $Q$ is the weight of vertices.
    $x^*(y) = y$, i.e. the pseudo-target of a vertex is itself, so $P(x) = Q(x)$.
    Condition~\ref{cond:ties-finite-x} is satisfied since $|x^*(y)| = 1$ for all $y \in \Y$, and $\X^* = \Y$.
    Then, the utility of $i$ is \[u_i(x,\bm{x}_{-i})=\sum_{y \in \Y} Q(y) \left[\frac{\mathbbm{1}_{d(x,y)\leq d(x',y) \;\forall x\neq x'}}{|\arg\min_{j\neq i} d(x_j,y)|}\right].\]
    (Note our model easily extends to settings where edges are weighted; $G$ is a hypergraph; or $\X$ is a richer space, with $V \subseteq \X$.)
\end{enumerate}

\section{Proofs omitted from \S~\ref{sec:pure}}\label{appendix:pure-equilibria}

\begin{subsection}{Proof of Lemma~\ref{obs:pure-utility-bound}}
\end{subsection}
\begin{proof}
We have
    \begin{align*}
        u_i(x_i, \bm{x}_{-i}) &=\E_{y \sim Q}
        \frac{\mathbbm{1}_{x_i\in X_{\min}(\bm{x}, y)}}{| X_{\min}(\bm{x}, y)|} \\
        &\geq \frac{Q(\{y \in \Y: x^*(y) = x_i\})}{|X_{\min}(\bm{x}, y)|} \\
        &= \frac{P(x_i)}{\countv(x_i)}.
    \end{align*}
    The last equality follows by definition of $P(x_i)$, and the fact that $Q(\{y: |x^*(y)| > 1\}) = 0$; thus $\E_{y \sim Q} [\X_{\min}(\bm{x}, y)] = \{x_j \in \bm{x}: x_j = x_i\}$.
    The statement holds with equality if $\countv(x) \geq 1$ for all $x \in \X^*$ (in which case $\mathbbm{1}_{x_i\in X_{\min}(\bm{x}, y)} = 0$ for $\hat y \notin \{y \in \Y: x^*(y) = x_i\}$).
    
    An analogous argument holds to show that $u_i(x', \bm{x}_{-i}) \geq \frac{P(x')}{\countv(x')+1}$.
    In this case, note that if $\countv(x_i) \geq 2$ and $\bm{x}$ covers $\X^*$, all pseudo-targets $x \in \X^*$ are still covered after the deviation (in which case $\mathbbm{1}_{x'\in X_{\min}(\bm{x}, y')} = 0$ for $\hat y \notin \{y \in \Y: x^*(y) = x'\}$) and equality holds.
    
\end{proof}

\begin{subsection}{Proof of Lemma \ref{lemma:no-pure-nash-2/n}}
\begin{proof}
    Consider a version of the game with $\X = \Y = \{x_1, x_2, x_3\}$ where $d(x_i, x_i) < d(x_j, x_i)$ for any $j \neq i$. It follows the pseudo-targets correspond to $\{x_1,x_2,x_3\}$.
    Take any $\epsilon \in (0, 1/2]$ and $n > \max \{6, \frac{4}{\epsilon} - 4\}.$
    Let $Q(y_1) = P(x_1) = \frac{2 - 2\epsilon}{n},$ $Q(y_2) = P(x_2) = \frac{2 - \epsilon}{n}$, and $Q(y_3) = P(x_3) = \frac{n - 4 + 3\epsilon}{n},$ so that $\pmin = \frac{2 - \epsilon}{n}.$

    We start by claiming that the pseudo-targets are covered in any pure equilibrium.
    If $\epsilon < \frac{1}{2}$, then $n > \frac{1}{\pmin}$ and the statement holds immediately by Theorem~\ref{thm:pure-extreme}.
    If $\epsilon = \frac{1}{2}$, we have $P(x_1) = \frac{1}{n}$, $P(x_2) = \frac{1.5}{n}$, and $P(x_3) = \frac{n - 2.5}{n}.$ 
    It follows that $x_2$ and $x_3$ must have at least one agent on each of them: otherwise, by Lemma~\ref{lem:lower-bound-util}, there exists an agent with utility $u_i(x_i, \bm{x}_{-i}) \leq 1/n$ who can deviate to either $x_2$ or $x_3$ and gain utility $> 1/n.$
    
    Meanwhile, we can also show $\countv(x_1) \geq 1$: assume not, so that $\countv(x_1) = 0$, $\countv(x_2) = n - \ell$, and $\countv(x_3) = \ell$ for some integer $\ell$ with $1 \leq \ell \leq n-1.$ 
    Then the utility of any agent $i$ on $x_2$ is the total probability share of $x_1$ and $x_2$, i.e. $u_i(x_2, \bm{x}_{-i}) = \frac{2.5}{n(n-\ell)}$, while the utility of any agent $j$ on $x_3$ is $u_j(x_3, \bm{x}_{-i}) = \frac{2.5}{n\ell}.$
    If either $u_i(x_2, \bm{x}_{-i})$ or $u_j(x_3, \bm{x}_{-i})$ are strictly less than $1/n$, there exists an agent with a profitable deviation to $x_1$ with utility $1/n$. 
    Assume not, so $u_i(x_2, \bm{x}_{-i}) \geq 1/n$ and $u_j(x_3, \bm{x}_{-i}) \geq 1/n$.
    Then we must have $\ell \geq n - 2.5$ and $\ell \leq n - 2.5$, or $\ell = n - 2.5$: but since $\ell$ and $n$ are integers, this is impossible.

    Now that we have shown for any $\epsilon \in (0, 1/2]$ that the pseudo-targets are covered, we also note that both $x_1$ and $x_2$ must have at most two agents on them: assume otherwise, so there are at least three agents on $x_1$ or $x_2$, with utility at most $\frac{2 - \epsilon}{3n}$ by Lemma~\ref{obs:pure-utility-bound}.
    Then, by Lemma~\ref{obs:pure-utility-bound}, one of them can jump to $x_3$ and gain utility at least $\frac{n - 4 + 3\epsilon}{n(n-3)} > \frac{2-\epsilon}{3n}$ (this inequality holds for any $\epsilon$ when $n > 6$). 
    Thus we have four possibilities for a pure Nash equilibrium profile of $(\countv(x_1), \countv(x_2), \countv(x_3))$; we consider each case and show that it is not a pure equilibrium.
    \begin{enumerate}
        \item $(1, 1, n-2):$ agent $i$ on $x_1$ has a profitable deviation to $x_2$. 
        Specifically, $u_i(x_i, \bm{x}_{-i}) = \frac{2 - 2\epsilon}{n}$, while deviating to $x_2$ leads to utility $\frac{2 - \epsilon}{2n} + \frac{2 - 2\epsilon}{2n} = \frac{4 - 3 \epsilon}{2n} > \frac{2 - 2\epsilon}{n}.$ (That is, deviating to $x_2$ allows $i$ to share the utility of both $x_1$ and $x_2$.)
        \item $(1, 2, n-3): $ an agent $i$ on $x_2$ has a profitable deviation to $x_3$. 
        Specifically, $u_i(x_i, \bm{x}_{-i}) = \frac{2-\epsilon}{2n}$, and deviating to $x_3$ leads to utility $\frac{n-4+3\epsilon}{n(n-2)} > \frac{2-\epsilon}{2n}$ for $\epsilon > \frac{4}{n+4}$.
        \item $(2, 1, n-3): $ agent $i$ on $x_1$ has a profitable deviation to $x_2$. 
        Specifically, $u_i(x_i, \bm{x}_{-i}) = \frac{2 - 2\epsilon}{2n}$, while deviating to $x_2$ leads to utility $\frac{2 - \epsilon}{2n} > \frac{2 - 2\epsilon}{2n}.$
        \item $(2, 2, n-4):$ an agent $i$ on $x_1$ has a profitable deviation to $x_3$. 
        Specifically, $u_i(x_i, \bm{x}_{-i}) = \frac{1-\epsilon}{n}$, and deviating to $x_3$ leads to utility $\frac{n-4+3\epsilon}{n(n-3)} > \frac{1-\epsilon}{n}$ for $\epsilon > \frac{1}{n}$.
    \end{enumerate}
\end{proof}
\end{subsection}

% \begin{subsection}{Proof of Lemma \ref{lemma:no-pure-nash}}

% \begin{proof}
%     Consider the Hotelling game with $\X = \{1, 2, 3\}$, $n = 2$, and consumer masses $\lambda_{\pi_1} = \lambda_{\pi_2} = \lambda_{\pi_3} = \frac{1}{3}$ for strict rankings $\pi_1 = (1, 3, 2), \pi_2 = (3, 2, 1),$ and $\pi_3 = (2, 1, 3)$.
%     It follows $P(1) = P(2) = P(3) = \frac{1}{3},$ so that $n = 2 < \frac{1}{\pmin}.$
%     We prove all strategy profiles $(x_1, x_2)$ are not equilibria in cases:
%     \begin{enumerate}
%         \item $x_1 = x_2 = x$ for some $x \in \{1, 2, 3\}.$
%         It follows that the retailers both have utility share $\frac{1}{\countv(x)} \left(\lambda_{\pi_1} + \lambda_{\pi_2} + \lambda_{\pi_3}\right) = \frac{1}{2}$. 
%         Then one can check that there is a location $x' \neq x$ that is ranked before $x$ for two of the three rankings, so either agent can deviate to this location $x'$ and gain utility $\frac{2}{3}.$
%         \item $x_1 \neq x_2.$
%         Then by construction of $\pi_1, \pi_2,$ and $\pi_3$, one of the retailers wins against the other for only one of the three rankings, with utility share $\frac{1}{3}$.
%         WLOG, assume this is retailer $1$. 
%         Retailer $1$ then has a profitable deviation to $x_2$, under which they recover utility $\frac{1}{2}.$
%     \end{enumerate}
% \end{proof}
% \end{subsection}

\begin{subsection}{Proof of Theorem \ref{thm:pure-existence}}

\begin{proof}
    We aim to prove that the strategy profile returned by Algorithm~\ref{alg:pure-construction} is an equilibrium.
    Note that for any $n \geq \frac{2}{\pmin}$, the value of $n_0$  in Algorithm~\ref{alg:pure-construction} satisfies
    \[n_0 = \sum_{x \in \X^*} \left\lfloor \frac{2 P(x)}{\pmin} \right\rfloor \leq 2 \sum_{x \in \X^*} \frac{P(x)}{\pmin} = \frac{2}{\pmin} \leq n.\]
    Thus it suffices to prove that the strategy profile $\bm{x}^{(n)}$ constructed by Algorithm~\ref{alg:pure-construction} for input $n$ is an equilibrium by induction over $n \geq n_0$.
    \
    Consider the base case $n_0$. 
    Then $\bm{x}^{(n_0)}$ is extreme by construction, and for all $x \in \X^*$, $k_{n_0}(x) = \left\lfloor \frac{2 P(x)}{\pmin} \right\rfloor \geq \left\lfloor \frac{2 P(x)}{P(x)} \right\rfloor = 2$. 
    By Lemma~\ref{lem:covered-deviations}, then, any potential best response $x'$ to the strategy profile $\bm{x}^{(n_0)}$ generated by Algorithm~\ref{alg:pure-construction} must satisfy $x' \in \X^*$.

    For ease of notation, let $x_i := x_i^{(n_0)}$.
    By Lemma~\ref{obs:pure-utility-bound}, since $k_{n_0}(x_i) \geq 2$ we have for any potential deviation $x' \neq x_i$, $x' \in \X^*$ that
    \[ 
            u_i(x_i, \bm{x}_{-i}^{(n_0)}) = \frac{P(x_i)}{k_{n_0}(x_i)} = \frac{P(x_i)}{ \left\lfloor \frac{2 P(x_i)}{\pmin} \right\rfloor} \geq \frac{\pmin}{2} > \frac{P(x')}{\left\lfloor \frac{2 P(x')}{\pmin} \right\rfloor + 1} = u_i(x', \bm{x}_{-i}^{(n_0)}).
    \]
    By definition, then, $\bm{x}^{(n_0)}$ is an equilibrium.

    Now suppose we have an equilibrium $\bm{x}^{(n)}$ with assignments $k_n$ for any $n \geq n_0$.
    Algorithm~\ref{alg:pure-construction} constructs a strategy profile $\bm{x}^{(n+1)} = \bm{x}^{(n)} \cup \{x_t^*\}$ for $t=n+1$, i.e. agent $t$ is added to the equilibrium profile $\bm{x}^{(n)}$ with reported position $x_t^*.$
    We aim to prove that $\bm{x}^{(n+1)}$ is an equilibrium.
    Note that $k_{n+1}(x_t^*) = k_n(x_t^*) + 1$ and $k_n(x) = k_{n+1}(x)$ for all $x \neq x_t^*$.
    It follows by Lemma~\ref{obs:pure-utility-bound} that
    \[ 
            u_t(x_t^*, \bm{x}_{-t}^{(n+1)}) = \frac{P(x_t^*)}{k_{n+1}(x_t^*)} = \frac{P(x_t^*)}{k_n(x_t^*) + 1} \geq \frac{P(x)}{k_n(x) + 1} = \frac{P(x)}{k_{n+1}(x) + 1} = u_t(x, \bm{x}_{-t}^{(n+1)}),
    \]
    where the inequality holds by the algorithm's selection of $x_t^*$ in line~\ref{alg:argmax}.
    By symmetry, the inequality also holds for any other agent $j \neq i$ that plays $x_t^*$, so that they also do not have a profitable deviation.
    
    Moreover, for any other agent $j \neq i$ with $x_j \neq x_t^*$ and potential deviation $x' \neq x_j$, $x' \in \X^*$,
    \[ 
            u_j(x_j, \bm{x}_{-j}^{(n+1)}) = \frac{P(x_j)}{k_{n+1}(x_j)} = \frac{P(x_j)}{k_n(x_j)} \geq \frac{P(x')}{k_n(x') + 1} \geq \frac{P(x')}{k_{n+1}(x') + 1} = u_j(x', \bm{x}_{-j}^{(n+1)}).
    \]
    The first inequality holds by the inductive hypothesis since $u_j(x_j, \bm{x}_{-j}^{(n)}) = \frac{P(x_j)}{k_n(x_j)}$ and $u_j(x', \bm{x}_{-j}^{(n)}) = \frac{P(x')}{k_n(x') + 1}$ under equilibrium profile $\bm{x}^{(n)}.$
    By definition, then, $\bm{x}^{(n+1)}$ is an equilibrium.

    % We now prove the second statement: that if $n > \frac{1}{\pmin},$ every pure Nash equilibrium is extreme and covers $\X^*$.
    % Note this assumption implies that $P(x) > \frac{1}{n}$ for all $x \in \X^*$.
    % We first show that any equilibrium $\bm{x}$ covers $\X^*$, i.e. for all $x \in \X^*$, $\countv(x) \geq 1$.
    % Assume not, so that there is some $\hat x \in \X^*$ with $\countv(\hat x) = 0$.
    % Then by Observation~\ref{obs:agent-low-util}, there exists some agent $i$ with utility $u_i(x_i, \bm{x}_{-i}) \leq \frac{1}{n}$.
    % By Lemma~\ref{obs:pure-utility-bound}, agent $i$ can deviate to $\hat x$ and gain utility $u_i(\hat x, \bm{x}_{-i}) \geq \frac{P(\hat x)}{\countv(\hat x) + 1} = P(\hat x) > \frac{1}{n}$.

    We now prove the second statement: that if $n > \frac{2}{\pmin},$ in any pure Nash equilibrium, $\countv(x) \geq 2$ for each $x \in \X^*$.
    Assume not, so that there is some $\hat x \in \X^*$ with $\countv(\hat x) \leq 1$.
    Then by Observation~\ref{obs:agent-low-util}, there exists some agent $i$ with utility $u_i(x_i, \bm{x}_{-i}) \leq \frac{1}{n}$.
    By Lemma~\ref{obs:pure-utility-bound}, agent $i$ can deviate to $\hat x$ and gain utility $u_i(\hat x, \bm{x}_{-i}) \geq \frac{P(\hat x)}{\countv(\hat x) + 1} \geq \frac{P(\hat x)}{2} > \frac{1}{n}$.
\end{proof}
\end{subsection}

\begin{subsection}{Proof of Lemma \ref{lem:pure-lower-bound}}

\begin{proof}
For contradiction, assume in an equilibrium $\bm{x}$ there is an pseudo-target $x \in \X^*$ such that $\countv(x) < \lfloor cn\rfloor$.
By Theorems~\ref{thm:pure-extreme} and~\ref{thm:pure-existence}, the equilibrium profile $\bm{x}$ is extreme and covers all $\X^*$, with at least two players on pseudo-target; and by Observation~\ref{obs:agent-low-util}, there exists some agent $i$ with reported position $x_i = \hat x$ such that $u_i(x_i, \bm{x}_{-i}) \leq \frac{1}{n}$.

Since $n \geq \frac{1}{c} > \frac{2}{\pmin}$, $P(x) > \frac{2}{c} \geq \frac{2}{n}$ for all $x \in \X^*.$
It follows by Lemma~\ref{lem:covered-deviations} that the utility of any agent $j$ playing $x$ (whose existence is guaranteed by Theorem~\ref{thm:pure-extreme}) is $u_j(x, \bm{x}_{-j}) = \frac{P(x)}{\countv(x)} > \frac{P(x)}{\lfloor cn\rfloor} \geq \frac{1}{n}$.
(Thus $x \neq \hat x.$)
If forecaster $i$ deviates to $x$, by Lemma~\ref{lem:covered-deviations} their utility will satisfy $u_i(x, \mathbf{x}_{-i}) = \frac{P(x)}{\countv(x) + 1} > \frac{c}{\lfloor cn\rfloor-1+1} = \frac{1}{n}$; thus $u_i(x, \bm{x}_{-i}) > \frac{1}{n} \geq u_i(x_i, \bm{x}_{-i})$.
Since forecaster $i$ has a strictly profitable deviation to $x$, we have reached a contradiction.
\end{proof}
\end{subsection}

\begin{subsection}{Proof of Lemma \ref{lem:wp-convergence}}

\begin{proof}
Consider any equilibrium profile $\bm{x}$ of $\game$.
For ease of notation, we let $u_{\bm{x}}(x)=\frac{P(x)}{\countv(x)}$ and  $\hat{u}_{\bm{x}}(x')=\frac{P(x')}{\countv(x')+1}$ for any $x' \neq x$, $x' \in \X^*$.
By Theorems~\ref{thm:pure-extreme} and~\ref{thm:pure-existence}, any equilibrium is extreme with at least two agents on each pseudo-target. 
Thus by Lemmas~\ref{obs:pure-utility-bound} and~\ref{lem:covered-deviations}, $u_i(x, x_{-i}) = u_{\bm{x}}(x)$ and for any potential deviation $x' \in \X^*$ for agent $i$, $u_i(x', x_{-i}) = \hat{u}_{\bm{x}}(x')$.

Consider any two pseudo-targets $x, x' \in \X^*$.
Then in equilibrium there exist two agents $i, j$ with $x_i = x$ and $x_j = x'$ by Lemma~\ref{lem:pure-lower-bound}.
For $\bm{x}$ to be an equilibrium, we must have $u_{\bm{x}}(x) \geq \hat{u}_{\bm{x}}(x') = \frac{P(x')}{\countv(x') + 1}$ and $u_{\bm{x}}(x') \geq \hat{u}_{\bm{x}}(x) = \frac{P(x)}{\countv(x) + 1}$.
Thus for all $x \neq x'$,
\begin{align}
\frac{u_{\bm{x}}(x)}{u_{\bm{x}}(x')} &\leq \frac{P(x)}{\countv(x)} \cdot \frac{\countv(x) + 1}{P(x)} \label{eq:upper-bound-wp} \\
&= \frac{\countv(x)+1}{\countv(x)} \nonumber \\
&\leq \frac{\lfloor c n \rfloor + 1}{\lfloor c n \rfloor}  \nonumber,
\end{align}
where the last inequality follows from Lemma~\ref{lem:pure-lower-bound}.
We can use the same argument to show that for all $x \neq x'$,
\begin{align}
\frac{u_{\bm{x}}(x')}{u_{\bm{x}}(x)} &\geq \frac{\countv(x)}{P(x)} \cdot \frac{P(x)}{\countv(x) + 1} \label{eq:lower-bound-wp} \\
&= \frac{\countv(x)}{\countv(x)+1} \nonumber \\
&\geq \frac{\lfloor c n \rfloor}{\lfloor c n \rfloor + 1} \nonumber.
\end{align}
where the last inequality follows since $\countv(x) \geq \lfloor c n \rfloor$ by Lemma~\ref{lem:pure-lower-bound}, and the function $f(z) = \frac{z}{z+1}$ is increasing for $z \geq 1.$
Now note that by Observation~\ref{obs:agent-low-util}, there exists some $x'$ such that $u_{\bm{x}}(x') \leq \frac{1}{n}$, and some $x$ such that $u_{\bm{x}}(x) \geq \frac{1}{n}.$
The result follows by plugging these bounds into Inequalities~\ref{eq:upper-bound-wp} and~\ref{eq:lower-bound-wp} for $u_{\bm{x}}(x')$ and $u_{\bm{x}}(x)$.

\end{proof}
\end{subsection}

\begin{subsection}{Proof of Proposition~\ref{prop:pure-non-extreme}}

\begin{proof}
    Consider the forecasting setting with $\X = [0,1]$, $\Y = \X^* = \{0,1\}$, $P(0) = P(1) = \frac{1}{2}$ and $n = \frac{1}{\pmin} = 2$.
    Then we show it is an equilibrium for both agents $i$ and $j$ to play $x_i = x_j = \frac{1}{2}.$
    Note first that $u_i(x_i, \bm{x}_{-i}) = \frac{1}{n}$ for both agents $i$.
    WLOG, consider one of the agents $i$ deviating to a point $x' \in (\tfrac{1}{2}, 1].$
    Note that $\mathbbm{1}_{x'\in X_{\min}((x', x_j), 0)} = 0$ since the other player is closer to $x^*(0),$ and $| X_{\min}((x', x_j), 1)| = 1$ since $x'$ is closer to $x^*(1)$.
    Then that agent $i$ will have utility share $\frac{1}{2}$, since
    \[
    u_i(x', \bm{x}_{-i}) = \E_{y \sim Q} \frac{\mathbbm{1}_{x'\in X_{\min}((x', x_j), y)}}{| X_{\min}((x', x_j), y)|} = P(1) = \frac{1}{2}.
    \]
    A symmetric argument holds if $x' \in [0, \tfrac{1}{2}).$
    Thus $\bm{x} = (\tfrac{1}{2}, \tfrac{1}{2})$ is an equilibrium, which is not extreme and does not cover $\X^*$.
\end{proof}
\end{subsection}

\section{Proofs omitted from \S~\ref{sec:mixed-equilibria}} \label{appendix:mixed-equilibria}

\subsection{Existence of equilibria}

We prove that for any symmetric, constant-sum game with pure strategy set corresponding to a compact Hausdorff space, a mixed symmetric equilibrium exists.
This immediately implies Theorem~\ref{thm:mixed-existence}.

\begin{observation} \label{obs:constant-symmetric-game}
    Consider a symmetric game that is constant-sum for some constant $c$.
    Then the utility $u_i$ of any agent $i$ over any symmetric mixed strategy profile $\sigma$ is $u_i(\sigma, \bm{\sigma}_{-i}) = \frac{c}{n}$.
\end{observation}

\begin{theorem}
Any symmetric, constant-sum game with pure strategy set corresponding to a compact Hausdorff space has a mixed symmetric equilibrium. 
\end{theorem}

\begin{proof}
    By Observation~\ref{obs:constant-symmetric-game}, the utility function $u_i(\sigma) = u_i(\sigma, \bm{\sigma}_{-i}) = \frac{c}{n}$ over symmetric strategy $\sigma$ is constant, and thus continuous. 
    Now, fix any agent $i$. 
    For any mixed symmetric strategy $\sigma$ and every $\epsilon > 0$, if all other agents deviate on the diagonal to some $\sigma'$, agent $i$ can secure utility $u_i(\hat \sigma, \bm{\sigma'}_{-i}) = \frac{c}{n} = u_i(\sigma, \bm{\sigma}_{-i})$ by deviating to $\hat \sigma = \sigma'$.
    By definition, then, the game is payoff secure (defined in~\citet{reny1999existence}), so all the conditions in Corollary 5.3 of~\citet{reny1999existence} are met and there exists a mixed symmetric equilibrium.
\end{proof}

\subsection{Characterization of equilibria}

\begin{lemma} \label{lem:deviation-lower-bound}
    For any mixed symmetric strategy profile $\bm{\sigma}$, where $\sigma_x \in (0, 1]$, the utility of deviating to any position $x \in X^*$  satisfies 
    \begin{equation}
        u_i(x, \bm{\sigma}_{-i}) \geq \frac{p_x}{n\sigma_x}(1-\sigma_x^n-(1-\sigma_x)^n)+\frac{1}{n}\sigma_x^{n-1}; 
    \end{equation}
    and when $\sigma_x=0$, $u_i(x,\bm{\sigma}_{-i})=p_x$.
\end{lemma}

\begin{proof}
    Let agent $i$ play pseudo-position $x \in \X^*.$
    We compute the utility conditioned on the events $E_1,E_2,$ and $E_3$ defined in \S\ref{subs:events}.
    Formally, note that by the law of total expectation,
    \[u_i(x, \bm{\sigma}_{-i}) = \sum_{\ell} \Pr[E_{\ell}] \E_{\substack{ y \sim Q \\ x_{-i} \sim \sigma_{-i}}} \left[ \pi_i(x_i, \bm{x}_{-i}, y) \mathrel{\big|} E_l \right]. \]
    We calculate each term in the sum for the exhaustive events below. 
    \begin{enumerate}
        \item Consider event $E_1$.
        Then $\mathbbm{1}_{x\in X_{\min}(x,y)} = 1$,  and  when $k$ other agents are on $x$, since $Q(\{y: x^*(y)| > 1\}) = 0$, $|X_{\min}(\bm{x},y)| = k+1$.
        Let $B(n; k, p)$ denote the binomial probability mass function. 
        It follows that
        \begin{align}
            \Pr[E_1] \E_{\substack{ y \sim Q \\ x_{-i} \sim \sigma_{-i}}} \left[ \frac{\mathbbm{1}_{x\in X_{\min}(x,y)}}{|X_{\min}(x,y)|} \; \mathrel{\Big|} \;  E_1 \right] &= P(x) \sum_{k=0}^{n-2} \frac{B(k; n-1, \sigma_x)}{k+1}, \label{eq:e1}
        \end{align}
        where we use that $\Pr[x^*(Y) = x] = P(x)$ and $\E [ \frac{1}{|X_{\min}(\bm{x},y)|} \mid E_1] = \frac{\sum_{k=0}^{n-2} \frac{B(k;n-1, \sigma_x)}{k+1}}{\Pr[E_1]}$, since the probability that $k$ of $n-1$ agents draw $x$ is $B(k; n-1, \sigma_x)$.
       \item For the second event, we have
       \begin{equation} \label{eq:e2}
           \Pr[E_2] \E \left[\pi_i(x_i, \bm{x}_{-i}, y)) \mid E_2 \right] = \frac{\sigma_x^{n-1}}{n},
       \end{equation}
       because the probability $n-1$ agents play $x$ is $\sigma_x^n$, and the utility of agent $i$ conditioned on everyone else playing $x$ is $\frac{1}{n}.$
       \item We simply lower bound the last term: 
       \begin{equation} \label{eq:e3}
           \Pr[E_3] \E[ \pi_i(x_i, \bm{x}_{-i}, y) \mid E_3] \geq 0.
       \end{equation}
    \end{enumerate}
    
    We now simplify Equation~\ref{eq:e1} to explicitly construct the expression for a lower bound on $u_i(x,\bm{\sigma}_{-i})$, which excludes the third term.
    We have
    \begin{align}
        p_x \sum_{k=0}^{n-2} \frac{B(k; n-1, \sigma_x)}{k+1} &= p_x\left(\sum_{k=0}^{n-2}\frac{\binom{n-1}{k}\sigma_x^k(1-\sigma_x)^{n-1-k}}{k+1}\right) \nonumber \\
        &= \frac{p_x}{n}\sum_{k=0}^{n-2}\binom{n}{k+1}\sigma_x^k(1-\sigma_x)^{n-(k+1)} && \text{ (since $\tfrac{1}{k+1}\tbinom{n-1}{k} = \tfrac{1}{n}\tbinom{n}{k+1}$}) \nonumber \\
        &= \frac{p_x}{n\sigma_x}\sum_{k=0}^{n-2}\binom{n}{k+1}\sigma_x^{k+1}(1-\sigma_x)^{n-(k+1)} \nonumber \\
        &= \frac{p_x}{n\sigma_x}(1-\sigma_x^n-(1-\sigma_x)^n) && \text{(by the binomial theorem.)} \label{eq:final-e1}
    \end{align}
    
    % $\frac{1}{k+1}\binom{n-1}{k} = \frac{1}{k+1}\frac{(n-1)!}{k!(n-1-k)!}$ can be rewritten as $\frac{1}{k+1}\frac{(n-1)!}{k!(n-1-k)!}$; then, rearranging the denominator to $\frac{(n-1)!}{k!(n-(k+1)!}$ we absorb the first term into the second, giving us$\frac{(n-1)!}{(k+1)!(n-(k+1)!)}$.
    % Factoring out $\frac{1}{n}$ from this quantity then gives us $\frac{1}{n}\frac{n!}{(k+1)!(n-(k+1)!}$, i.e., $\frac{1}{n}\binom{n}{k+1}$, to obtain
    % \[\frac{p_x}{n}\sum_{k=0}^{n-2}\binom{n}{k+1}\sigma_x^k(1-\sigma_x)^{n-(k+1)}.\]
    Now, combining Equation~\ref{eq:final-e1} with Equation~\ref{eq:e2} and the lower bound in Inequality~\ref{eq:e3}, we obtain
    \[u_i(x,\bm{\sigma}_{-i})\geq\frac{p_x}{n\sigma_x}(1-\sigma_x^n-(1-\sigma_x)^n)+\frac{1}{n}\sigma_x^{n-1}.\]

    When $\sigma_x=0$, take $\lim_{\sigma_x\rightarrow0}u(x,\bm{\sigma}_{-i})$; applying L'Hopital's rule once gives $\lim_{\sigma_x\rightarrow0}u(x,\bm{\sigma}_{-i})=p_x$.
\end{proof}

\begin{lemma}\label{lem:G-monotonicity}
    Let $G(\sigma_x)$ be defined as in Equation \ref{eq:G}.
    Then $G(\sigma_x)$ is continuously differentiable and monotone increasing for $\sigma_x\in(0,1).$
\end{lemma}

\begin{proof}
    First, we show that $G(\sigma_x)$ is continuously differentiable for $\sigma_x\in(0,1)$.
    We calculate $G'(\sigma_x)$ as
    \[G'(\sigma_x)=\frac{1-n\sigma_x^{n-1}}{1-(1-\sigma_x)^n-\sigma_x^n}-\frac{\left(n(1-\sigma_x)^{n-1}-n\sigma_x^{n-1}\right)(\sigma_x-\sigma_x^n)}{\left((1-(1-\sigma_x)^n-\sigma_x^n)\right)^2}.\]
    
    Observe that in both terms, the denominator is a continuous function for all $\sigma_x\in(0,1)$.
    
    In the first term, the numerator is continuous for all $\sigma_x\in(0,1)$.
    Then, the quotient of two continuous functions is continuous for all values of $\sigma_x$ where the divisor is not 0; thus, the first term is a continuous function for $\sigma_x\in(0,1)$.
    In the second term, the numerator is the product of two continuous functions; the second term is a continuous function for $\sigma_x\in(0,1)$ for similar reasons to the first term.
    Finally, the difference of two continuous functions is continuous; thus, $G'(\sigma_x)$ is continuous.

    Next, we show that $G(\sigma_x)$ is monotone increasing for $\sigma_x\in(0,1)$.
    Let $G_{\text{num}}'(\sigma_x)=1-n\sigma_x^{n-1}$ be the derivative of the numerator of $G(\sigma_x)$, and let $G_{\text{den}}'(\sigma_x)=n\left((1-\sigma_x)^{n-1}-\sigma^{n-1}\right)$ be the derivative of the denominator of $G(\sigma_x)$.
    For $\sigma_x\in\left(0,(\frac{1}{n})^{\frac{1}{n-1}}\right)$, $G_{\text{num}}'(\sigma_x)>0$, and the denominator of $G$ is less than 1, so $G(\sigma_x)$ is increasing for $\sigma_x\in\left(0,(\frac{1}{n})^{\frac{1}{n-1}}\right)$.
    Then, for $\sigma_x\geq(\frac{1}{n})^{\frac{1}{n-1}}$, observe that both $G_{\text{den}}'(\sigma_x)$ and $G_{\text{num}}'(\sigma_x)$ are negative, and $G_{\text{den}}'(\sigma_x)<G_{\text{num}}'(\sigma_x)$.
    It follows that $G(\sigma_x)$ is increasing for $\sigma\in(0,1)$.
\end{proof}

\begin{lemma} \label{lem:G-equality}
        Let $G(\sigma_x)$ be defined as in Equation \ref{eq:G}. 
        Then $G^{-1}(p_x)$ exists for $p_x \in \left[\frac{1}{n}, 1 - \frac{1}{n}\right],$ and 
        \[\sup_{p_x \in \left[\frac{1}{n}, 1 - \frac{1}{n}\right]} |G^{-1}(p_x) - p_x| = \frac{1}{n}. \]
\end{lemma}
\begin{proof}
    Let
    \begin{equation}\label{eq:F}
        F(p_x,\sigma_x)=\frac{\sigma_x-\sigma_x^n}{1-\sigma_x^n-(1-\sigma_x)^n}-p_x.
    \end{equation}
    By the implicit function theorem, for $F(p_x,\sigma_x)=0$ continuously differentiable, there exist some $\hat{\sigma}_x$ and $\hat{p}_x$ such that: if (1) $F(\hat{p}_x,\hat{\sigma}_x)=0$ and (2) $\frac{\partial F}{\partial p_x}\neq 0$, there exists a unique differentiable function $\phi(p_x)$ for which $\phi(p_x)=\sigma_x$ and $F(p_x,\phi(p_x))=0$.
    
    First, observe that $F(p_x,\sigma_x)$ is continuously differentiable for all $\sigma_x\in(0,1)$ and for all $p_x$, as $\frac{\partial F}{\partial\sigma_x}$ exists and is continuous by Lemma \ref{lem:G-monotonicity}, and $\frac{\partial F}{\partial p_x}$ is clearly continuous and nonzero for all $p_x$.
    Next, observe that $F(p_x,\sigma_x)=0$ when $\frac{\sigma_x-\sigma_x^n}{1-\sigma_x^n-(1-\sigma_x)^n}=p_x$.
    Hence, there exist $p_x$ and $\sigma_x$ such that $F(p_x,\sigma_x)=0$.
    Thus, the conditions for existence of $\phi(p_x)$ are satisfied.
        
    Take $G(\sigma_x)$ to be the left-hand side of the above; then, we can see that $\phi(p_x)=G^{-1}(p_x)$.
    Thus, $G^{-1}(\hat{p}_x)$ gives the equilibrium $\hat{\sigma}_x$ for some $\hat{p}_x$.

    Next, for convenience, we show that $G(\sigma_x)$ is odd and symmetric about $(\frac{1}{2},\frac{1}{2})$, which will allow us to restrict our analysis to the region where $\sigma_x\in(\frac{1}{n},\frac{1}{2}]$.
    The translation of $G(\sigma_x)$ to the origin, $G_0(\sigma_x)=\frac{(\sigma_x+1/2)(\sigma_x+1/2)^n}{1-(1/2-\sigma_x)^n-(\sigma_x+1/2)^n}-\frac{1}{2}$, is odd.
    Hence $G(\sigma_x)$ is odd and symmetric with respect to $(\frac{1}{2},\frac{1}{2})$.
    The inverse of an odd function is odd, and so $G^{-1}$ is also odd, and $G^{-1}(p_x)=G(\sigma_x)=\frac{1}{2}$ for $\sigma_x=p_x=\frac{1}{2}$.
    Therefore, the supremum of $|G^{-1}(p_x)-p_x|$ for $p_x<\frac{1}{2}$ will be the same as for $p_x>\frac{1}{2}$.

    Now we compute $\sup_{p_x\in[\frac{1}{n},\frac{1}{2}]}|G^{-1}(p_x)-p_x|$ by evaluating $G^{-1}(p_x)-p_x$ at all its critical points.
    By the inverse function formula, $(G^{-1})'(p_x)=\frac{1}{G'(\sigma_x)}$, and so $(G^{-1})'(p_x)=1$ when $\frac{1}{G'(\sigma_x)}=1$, which occurs at $\sigma_x=p_x=\frac{1}{2}$, the unique point where $|G^{-1}(p_x)-p_x|=0$.
    The next critical point is the interval boundary, $p_x=\frac{1}{n}$.
    We seek the value $G^{-1}(\frac{1}{n})=G^{-1}(G(\hat{\sigma}_x))$ where $G(\hat{\sigma}_x)=\frac{1}{n}$; i.e., we seek a $\hat\sigma_x$ that solves 
    \[\frac{\hat\sigma_x-\hat\sigma_x^n}{1-(1-\hat\sigma_x)^n-\hat\sigma_x^n}=\frac{1}{n}\]
    Take $\hat\sigma_x=0$.
    $G(\sigma_x)$ is undefined at $\hat\sigma_x=0$; however, by applying l'Hopital's rule once, we can find that it limits to $\frac{1}{n}$ as $\hat\sigma_x\rightarrow0$.
    Thus, $\sup_{p_x\in[\frac{1}{n},1-\frac{1}{n}]}|G^{-1}(p_x)-p_x|=|0-\frac{1}{n}|=\frac{1}{n}$.
\end{proof}

\begin{lemma}\label{lem:high-p-equilibrium-2}
    For any $x \in \X^*,$ if $p_x \geq 1 - \frac{1}{n}$, there is a unique symmetric mixed equilibrium profile $\bm{\hat{\sigma}}=(\hat{\sigma},\ldots,\hat{\sigma})$ where everyone deterministically chooses position $x$, i.e. $\hat{\sigma}(x) = 1$, for all $i$.
\end{lemma}

\begin{proof}
    First, we prove that $\hat{\sigma}$ is an equilibrium.
    Assume not: then some agent $i$ has a best response $\sigma_i \neq \hat{\sigma}$ such that $u_i(\sigma_i, \hat{\sigma}_{-i}) > u_i(\hat{\sigma}, \ldots, \hat{\sigma}) = 1/n.$
    $\sigma_i(x) < 1$, so there exists some $x' \in \X$, $x' \neq x$ such that $\sigma_i(x') > 0.$
    Since everyone else plays $x$ with probability one, $u_i(x', \hat{\sigma}_{-i}) \leq 1 - p_x \leq \frac{1}{n}$. 
    But this contradicts indifference over actions in the support of the best response strategy $\hat{\sigma}_i.$
    Thus $\hat{\sigma}$ is an equilibrium.

    Now, we show that for any symmetric mixed strategy $\sigma$, if $u_i(x, \bm{\sigma}_{-i}) \leq \frac{1}{n}$ then $\sigma(x) \geq 1$.
    This proves by Lemma~\ref{obs:symmetric-equil-utility} that if $\sigma$ is a symmetric mixed equilibrium, $\sigma(x) \geq 1;$ otherwise, agent $i$ can deviate to playing $\hat{x}$ and achieve strictly better utility. 
    
    Note first that by Lemma~\ref{lem:deviation-lower-bound},
    \[u_i(x, \bm{\sigma}_{-i}) \geq \frac{p_x}{n\sigma_x}(1-\sigma_x^n-(1-\sigma_x)^n)+\frac{1}{n}\sigma_x^{n-1}, \]

    So it suffices to show that if
    \begin{align*}
        \frac{p_x}{n\sigma_x}(1-\sigma_x^n-(1-\sigma_x)^n)+\frac{1}{n}\sigma_x^{n-1} &\leq \frac{1}{n}\\
        p_x&\leq \frac{\sigma_x-\sigma_x^n}{1-\sigma_x^n-(1-\sigma_x)^n}\\
    p_x&\leq G(\sigma_x),
    \end{align*}
    then $\sigma_x \geq 1.$ 
    
    First note that by Lemma \ref{lem:G-monotonicity}, $G$ is continuous and monotone increasing over $(0,1)$ , so $G(\sigma_x) \geq p_x$ implies $\sigma_x \geq G^{-1}(p_x).$
    But for $p_x \geq 1 - \frac{1}{n},$ $G^{-1}(p_x) \geq 1$ because $G(\sigma_x)\geq1-\frac{1}{n}$ for $\sigma_x\geq1$, so $\sigma_x \geq G^{-1}(p_x) \geq 1$.
    (Note that $G^{-1}(p_x)$ is undefined when $G^{-1}(p_x)=1$; however, taking L'Hopital's rule once for $G(1-\frac{1}{n})$ we have that $G(1-\frac{1}{n})=1$, and the statement holds.)
\end{proof}

%\begin{lemma}\label{mixed-positive-measure}
%    For $n \geq 1$ and any symmetric mixed equilibrium $\sigma$, $x \in \supp(\sigma)$ for all $x \in \X^*$.
%\end{lemma}

% \begin{proof}
%     Assume not. 
%     Then there exists some $x \in \X^*$ such that $x \notin \supp(\sigma)$.
%     It follows any agent $i$ could deviate to playing $x$ and gain utility $u_i(x, \bm{\sigma}_{-i}) > \frac{1}{n}.$
% \end{proof}

\begin{proposition}\label{prop:mixed-lower-bound}
    Consider any $x \in \X^*$ induced by the game $\game$.
    Then in any symmetric mixed equilibrium $\sigma$, $\sigma_x \geq p_x - \frac{1}{n}$.
\end{proposition}

\begin{proof}
    Let $\sigma$ be some symmetric mixed strategy.
    We split into cases over $p_x$.
    
    \textbf{Case 1: } $p_x \leq \frac{1}{n}.$
    Then the statement holds trivially.
    
    \textbf{Case 2:} $p_x \in \left( \pmin, 1 - \frac{1}{n} \right).$
    Since $n\geq\frac{1}{\pmin}$, we must have $p_x \in \left( \frac{1}{n}, 1 - \frac{1}{n} \right).$
    We prove that for any $x \in \X^*,$ if $u_i(x, \bm{\sigma}_{-i}) \leq \frac{1}{n}$, then $\sigma_x \geq p_x - \frac{1}{n}$.
    It follows by Observation~\ref{obs:symmetric-equil-utility} that if $\sigma$ is a symmetric mixed equilibrium, $\sigma_x \geq p_x - \frac{1}{n}$: otherwise, $u_i(x,\bm{\sigma}_{-i})>\frac{1}{n}$, and player $i$ can deviate to $\sigma_x=1$ and achieve strictly better utility.

    %Note first by Lemma~\ref{mixed-positive-measure} that $\sigma_x > 0$.
    By Lemma~\ref{lem:deviation-lower-bound},
    \[u_i(x, \sigma_{-i}) \geq \frac{p_x}{n\sigma_x}(1-\sigma_x^n-(1-\sigma_x)^n)+\frac{1}{n}\sigma_x^{n-1}, \]

    so it suffices to show that if $G(\sigma_x) \geq p_x$ (as in Lemma \ref{lem:high-p-equilibrium-2}), then $\sigma_x \geq p_x - \frac{1}{n}$.
    
    %\begin{align*}
    %    \frac{p_x}{n\sigma_x}(1-\sigma_x^n-(1-\sigma_x)^n)+\frac{1}{n}\sigma_x^{n-1} &\leq \frac{1}{n}\\
    %    p_x&\leq \frac{\sigma_x-\sigma_x^n}{1-\sigma_x^n-(1-\sigma_x)^n}\\
    %    p_x&\leq G(\sigma),
    %\end{align*}
    
    By Lemma \ref{lem:G-monotonicity}, $G$ is continuously differentiable and monotone increasing over $(0,1)$, so $G(\sigma_x) \geq p_x$ implies $\sigma_x \geq G^{-1}(p_x).$
    But by Lemma~\ref{lem:G-equality}, $G^{-1}(p_x) \geq p_x - \frac{1}{n}$ for any $p_x \in \left[\frac{1}{n}, 1 - \frac{1}{n}\right]$, so $\sigma_x \geq p_x - \frac{1}{n}.$

    \textbf{Case 3:} $p_x \geq 1 - \frac{1}{n}.$ 
    Then by Lemma~\ref{lem:high-p-equilibrium-2}, the unique equilibrium satisfies $\sigma_x = 1$.
    Thus the statement trivially follows.
\end{proof}

\begin{lemma}\label{lem:mixed-extreme}
    Let $n > 8 \left( \frac{4}{\pmin} \log \frac{1}{\pmin} \right).$
    % Let $n \geq \frac{2}{\pmin}$, and $n > \frac{2(c+\log(|\X^*|))}{p_{\min}}$ for some constant $c$.
    Then in any symmetric mixed equilibrium $\sigma$, the probability $\sigma$ covers $\X^*$ is at least $1 - \frac{1}{n^2}.$
\end{lemma}

\begin{proof}
We begin by proving that if $n \geq N_0 = 8 \left( \frac{4}{\pmin} \log \frac{1}{\pmin} \right)$, then $n \geq \frac{2(\log(|\X^*| n^2)}{\pmin}.$
To do so, we break the proof into two cases: we show for all $n \geq N_0,$ (1) $\frac{2 \log(|\X^*|)}{\pmin} \leq \frac{n}{2}$, and (2) $\frac{4 \log n}{\pmin} \leq \frac{n}{2}$.
Then $n \geq \frac{4 \log n}{\pmin} + \frac{2 \log(|\X^*|)}{\pmin} = \frac{2(\log(|\X^*| n^2)}{\pmin}$, and the statement holds.

To show (1), we have $|\X^*| \leq \frac{1}{\pmin}$.
Thus $\frac{2 \log(|\X^*|)}{\pmin} \leq \left(\frac{2}{\pmin} \right) \log \left( \frac{1}{\pmin} \right) \leq N_0$.

To show (2), we must prove that $\log n \leq \frac{\pmin n}{8}.$
We note first that the function $\frac{4\log n}{\pmin} - \frac{n}{2}$ is decreasing whenever $n > \frac{8}{\pmin}$, which is satisfied for any $n \geq N_0.$
Let $a = \tfrac{\pmin}{8}$, so that $N_0 = \tfrac{4}{a} \log \tfrac{1}{a}$.
Thus it suffices to check that
\[
\log \left(\frac{4}{a}\log\frac{1}{a}\right) \leq 4 \log\frac{1}{a}.
\]

Letting $L = \log \tfrac{1}{a}$, this inequality becomes
\[
\log 4 + \log L \leq 3L,
\]
which holds for all $L \geq \log 8$, or $\pmin \in (0,1].$

The probability that no player plays some fixed pseudo-target $x\in\X^*$ is at most $(1-\sigma_x)^{n-1}$.
Taking a union bound over all $|\X^*|$ outcomes, and noting by Proposition~\ref{prop:mixed-lower-bound} that $\sigma_x\geq \pmin-\frac{1}{n}$ for all pseudo-targets $x$, the probability of any pseudo-target being unplayed is thus at most

\begin{align*}
&\sum_{x\in\X^*}\left(1-\left(\pmin-\frac{1}{n}\right)\right)^{n-1}\\
&\leq \sum_{x\in\X^*}e^{-(\pmin-\frac{1}{n})(n-1)}\\
&=|\X^*|e^{-(\pmin-\frac{1}{n})(n-1)} \\ 
&\leq |\X^*| e^{-(\pmin/2)(n-1)} \\
&\leq |\X^*| e^{-\log(|\X^*| n^2)} \\
&= \frac{1}{n^2}.
\end{align*}

The first inequality holds because $(1+x) \leq e^x$; the second inequality holds since $n \geq \frac{2}{\pmin}$; and the third inequality holds since $n \geq \frac{2(\log(|\X^*| n^2)}{\pmin}+1$.
\end{proof}

\begin{proposition} \label{prop:extreme-equilibria}
    If $n > 8 \left( \frac{4}{\pmin} \log \frac{1}{\pmin} \right),$ then all mixed symmetric equilibria are extreme. 
\end{proposition}

\begin{proof}
Assume $\sigma$ has support outside $\X^*$, i.e. $x' \in \supp(\sigma)$ for some $x' \notin \X^*$.
We note that if all pseudo-targets $x \in \X^*$ are covered, there is at least one agent playing $x^*(y)$ for each $y \in \Y$; in that case, $x'$ cannot be in the set $X_{\min}(\bm{x}, y)$ for any $y \in \Y,$ and the utility is 0.
But by Lemma \ref{lem:mixed-extreme}, the probability that any pseudo-target $x\in\X^*$ is uncovered is at most $\frac{1}{n^2}$.
It follows $u_i(x', \sigma_{-i}) \leq \frac{1}{n^2} < \frac{1}{n}$, contradicting that $x'$ is in the support of the mixed equilibrium.
\end{proof}

\begin{lemma}\label{lem:monotonicity-with-ties}
Let $\overline{G}(\sigma_x)=\frac{\sigma_x-\sigma_x^n-\sigma_x/n}{1-\sigma_x^n-(1-\sigma_x)^n}$.
Then, $\overline{G}(\sigma_x)$ is monotone increasing for $n\geq 13$ and $\sigma_x\in(0,1-\frac{1}{\sqrt{n}}+\frac{1}{n})$.
\end{lemma}

\begin{proof}
    %By Lemma \ref{lem:G-monotonicity}, we know that $G(\sigma_x)$ is monotone increasing; thus, $\overline{G}(\sigma_x)$ is monotone increasing for $\sigma_x\in(0,1-\frac{1}{\sqrt{n}}-\frac{1}{n})$ if $G'(\sigma_x)>\gamma'(\sigma_x)$ (equivalently, $G'(\sigma_x)-\gamma'(\sigma_x)>0$) for $\sigma_x\in(0,1-\frac{1}{\sqrt{n}}-\frac{1}{n})$.
    We first show that $\overline{G}(\sigma_x)$ is monotone increasing for $\sigma_x\leq 1-\frac{1}{\sqrt{n}}+\frac{1}{n}$.
    Let $\overline{G}_{\text{num}}'(\sigma_x)=1-n\sigma_x^{n-1}-\frac{1}{n}$ be the derivative of the numerator of $\overline{G}(\sigma_x)$, and observe that $\overline{G}'_{\text{den}}(\sigma_x)=G'_{\text{den}}(\sigma_x)$.
    %Then, by the same reasoning, $\overline{G}(\sigma_x)$ is increasing for $\sigma_x<\left(\frac{1}{n}\right)^{\frac{1}{n-1}}$.
    Then, for $\sigma_x>\left(\frac{1}{n}-\frac{1}{n^2}\right)^{\frac{1}{n-1}}$, $\overline{G}'_{\text{num}}(\sigma_x)>0$ and $0<\overline{G}'_{\text{den}}(\sigma_x)<1$; thus, $\overline{G}(\sigma_x)$ is increasing for $\sigma_x>(\frac{1}{n}-\frac{1}{n^2})^{\frac{1}{n-1}}$.
    %, $\overline{G}'(\sigma_x)>0$ and so the numerator is increasing while the denominator is less than 1.
    %and let $\overline{G}_{\text{den}}'(\sigma_x)=n\left((1-\sigma_x)^{n-1}-\sigma^{n-1}\right)$ be the derivative of the denominator of $\overline{G}(\sigma_x)$.

    %We have \[G'(\sigma_x)=\frac{1-n\sigma_x^{n-1}}{1-\sigma_x^n-(1-\sigma_x)^n}-\frac{(\sigma_x-\sigma_x^n)\left(n(1-\sigma_x)^{n-1}-n\sigma_x^{n-1}\right)}{\left(1-\sigma_x^n-(1-\sigma_x)^{n}\right)^2}\]

    %\[\gamma'(\sigma_x)=\frac{1/n}{1-\sigma_x^n-(1-\sigma_x)^n}-\frac{\sigma_x\left(n(1-\sigma_x)^{n-1}-n\sigma_x^{n-1}\right)}{n\left(1-\sigma_x^n-(1-\sigma_x)^n\right)^2}\]

    %Observe that the negative term of $\gamma'(\sigma_x)$ is more negative than that of $G'(\sigma_x)$ for $\sigma_x<1-\frac{1}{\sqrt{n}}+\frac{1}{n}$, as $\frac{\sigma_x}{n}>\sigma_x-\sigma_x^n$ for such $\sigma_x$.
    %Also, the positive term of $G'(\sigma_x)$, $\frac{1-n\sigma_x^{n-1}}{1-\sigma_x^n-(1-\sigma_x)^n}$, is greater than the positive term of $\gamma'(\sigma_x)$, i.e., for $\sigma_x<1-\frac{1}{\sqrt{n}}-\frac{1}{n}$, \[\frac{1-n\sigma_x^{n-1}}{1-\sigma_x^n-(1-\sigma_x)^n}>\frac{1/n}{1-\sigma_x^n-(1-\sigma_x)^n}\] since $1-n\sigma_x^{n-1}<\frac{1}{n}$ for such values of $\sigma_x$.
    For sufficiently large $n$ (in particular, $n\geq 13$), $(\frac{1}{n}-\frac{1}{n^2})^{\frac{1}{n-1}}>1-\frac{1}{\sqrt{n}}+\frac{1}{n}$.
    Thus, for such $n$, $\overline{G}(\sigma_x)$ is monotonically increasing for $\sigma_x\leq1-\frac{1}{\sqrt{n}}+\frac{1}{n}$.
\end{proof}

\begin{lemma} \label{lem:lower-bound-util}
    Let $n > 8 \left( \frac{4}{\pmin} \log \frac{1}{\pmin} \right).$ For $\sigma_x\in(0,1]$ and for all $x \in \X^*$, 
    \[u_i(x,\bm{\sigma}_{-i}) \leq \frac{p_x}{n\sigma_x}(1-\sigma_x^n-(1-\sigma_x)^n)+\frac{1}{n}\sigma_x^{n-1}+\frac{1}{n^2};\]
    when $\sigma_x=0$, $u_i(x,\bm{\sigma}_{-i})\leq p_x+\frac{1}{n^2}$.
\end{lemma}

\begin{proof}
    By Lemma \ref{lem:mixed-extreme}, the probability that at least one pseudo-target $x\in\X^*$ is uncovered is at most $\frac{1}{n^2}$.
    Then, we can use the same breakdown into events $E_1, E_2,$ and $E_3$ as Lemma~\ref{lem:deviation-lower-bound} to upper bound the utility of playing $x$.
    Specifically, the first two terms as derived in Lemma~\ref{lem:deviation-lower-bound} are exact.
    Meanwhile, $E_3$ is the event that $x^*(Y) \neq x$, and at least one other agent does not play $x.$
    In this case, if $\X^*$ is covered (which occurs with probability at least $1 - \frac{1}{n^2}$), agent $i$ gains utility 0 (since they are never in the set $X_{\min}(\bm{x}, y)$).
    Otherwise, we upper bound their utility by 1.
    It follows that
    \begin{equation}\label{eq:utility-upper-bound}
        u_i(x,\bm{\sigma}_{-i}) \leq \frac{p_x}{n\sigma_x}(1-\sigma_x^n-(1-\sigma_x)^n)+\frac{1}{n}\sigma_x^{n-1}+\frac{1}{n^2}.
    \end{equation}
    When $\sigma_x=0$, taking L'Hopital's rule once gives $u_i(x,\bm{\sigma}_{-i})=p_x+\frac{1}{n^2}$.
\end{proof}

\begin{proposition} \label{prop:mixed-upper-bound}
    Let $n > \max \left\{ 43, 8 \left(\frac{4}{\pmin} \log \frac{1}{\pmin} \right) \right\}$. Under any symmetric mixed equilibrium $\sigma$ and for all $x \in \X^*$, $\sigma_x \leq p_x + \frac{1}{n}$.
\end{proposition}

\begin{proof}
    Again, we break into cases:
    
    \textbf{Case 1.}
    $p_x\in\left[\pmin,1-\frac{1}{\sqrt{n}}\right]$.

    We will show that if $u_i(x, \bm{\sigma}_{-i}) \geq \frac{1}{n}$, then $\sigma_x \leq p_x + \frac{1}{n}$.
    Then $\sigma$ is a mixed Nash equilibrium if and only if $\sigma_x \leq p_x + \frac{1}{n}$ for all $x\in\X^*$; otherwise, $u_i(x,\bm{\sigma}_{-i})<\frac{1}{n}$ and thus $x$ could not be in the support of $\sigma$ by Lemma~\ref{obs:symmetric-equil-utility}.

    By Lemma~\ref{lem:lower-bound-util}, we know 
    \[
    u_i(x,\bm{\sigma}_{-i}) \leq \frac{p_x}{n\sigma_x}(1-\sigma_x^n-(1-\sigma_x)^n)+\frac{1}{n}\sigma_x^{n-1}+\frac{1}{n^2}.
    \]
    Moreover, by Lemma \ref{obs:symmetric-equil-utility}, we must have $u_i(x, \bm{\sigma}_{-i}) = \frac{1}{n}.$ 
    Thus, for any value of $p_x$, the equilibrium $\sigma_x$ must satisfy
    \[\frac{p_x}{n\sigma_x}(1-\sigma_x^n-(1-\sigma_x)^n)+\frac{1}{n}\sigma_x^{n-1}+\frac{1}{n^2}\geq\frac{1}{n}.\]
    Equivalently, 
    \begin{align}
        p_x &\geq\frac{\sigma_x-\sigma_x^n-\frac{\sigma_x}{n}}{1-\sigma_x^n-(1-\sigma_x)^n}, \; \text{or} \nonumber \\
         p_x &\geq \overline G(\sigma_x). \label{eq:upper-bound-p-g}
    \end{align}
    Define \[\overline{F}(p_x,\sigma_x)=\frac{\sigma_x-\sigma_x^n-\frac{\sigma_x}{n}}{1-\sigma_x^n-(1-\sigma_x)^n}-p_x.\]

    Observe that $\overline{G}(\sigma_x)$ is continuously differentiable and nonzero for $\sigma_x\in(0,1-\frac{1}{\sqrt{n}}+\frac{1}{n})$, because by Lemma \ref{lem:G-monotonicity}, $G(\sigma_x)$ is continuously differentiable and the difference of two continuously differentiable functions is also continuously differentiable.
    Further, there exists $\hat{p}_x$ and $\hat{\sigma}_x$ such that $\overline{F}(p_x,\sigma_x)=0$.
    Then, as in Lemma \ref{lem:G-equality}, by the implicit function theorem, $\overline{G}^{-1}(p_x)$ exists for $p_x$ such that $G^{-1}(p_x)\in(0,1-\frac{1}{\sqrt{n}})$ (i.e., for $\sigma_x\in(0,p_x+\frac{1}{n})$ where $p_x=1-\frac{1}{\sqrt{n}}$).

    We first show that we have $\overline{G}(p_x+\frac{1}{n})\geq p_x$:

    \begin{align*}
        \frac{(p_x+\frac{1}{n})-(p_x+\frac{1}{n})^{n-1}-\frac{p_x+1/n}{n}}{1-(p_x+\frac{1}{n})^n-(1-p_x-\frac{1}{n})^n}-p_x&>0\\
        \left(p_x+\frac{1}{n}\right)\left(1-\frac{1}{n}\right)-\left(p_x+\frac{1}{n}\right)^{n-1}&>p_x\left(1-\left(p_x+\frac{1}{n}\right)^n-\left(1-p_x-\frac{1}{n}\right)^n\right)\\
        \left(p_x+\frac{1}{n}\right)\left(1-\frac{1}{n}\right)-\left(p_x+\frac{1}{n}\right)^{n-1}&>p_x
    \end{align*}
    
    where the third inequality follows from observing that $1-\sigma_x^n-(1-\sigma_x)^n\leq1$ for all $\sigma_x\in(\pmin,1)$.
    
    %where $G(\sigma_x)\geq p_x$ follows from Lemma \ref{lem:G-monotonicity}.
    %The inequality holds because $1-\sigma_x^n-(1-\sigma_x)^n\leq1$ for $\sigma_x\in(\pmin,1)$, and so the sum of the rightmost two terms is less than 0.

    Simplifying gives us \[\frac{1-p_x}{n}-\frac{1}{n^2}-\left(p_x+\frac{1}{n}\right)^{n-1}>0.\]
    
    Taking $p_x=1-\frac{1}{\sqrt{n}}$, this inequality holds for sufficiently large $n$ (in particular, when $p_x=1-\frac{1}{\sqrt{n}}$, we require $n\geq43$).
    Moreover, we have that $\overline{G}(\sigma_x)-\sigma_x$ is a decreasing function of $\sigma_x$; thus, for all values of $p_x<1-\frac{1}{\sqrt{n}}$, $\overline{G}(p_x+\frac{1}{n})-p_x>0$ as well.

    Having shown that $\overline{G}(p_x+\frac{1}{n})\geq p_x$, then, by monotonicity of $\overline{G}(\sigma_x)$, we have that $p_x+\frac{1}{n}\geq\overline{G}^{-1}(p_x)$.
    Finally, by Equation \ref{eq:upper-bound-p-g}, we have $p_x\geq \overline{G}(\sigma_x)$; then, by monotonicity of $\overline{G}(\sigma_x)$, we have that $p_x+\frac{1}{n}\geq \overline{G}^{-1}(p_x)\geq\sigma_x$, which shows the result.

    \textbf{Case 2.} $p_x \geq 1-\frac{1}{\sqrt{n}}$.
    We show that if $p_x \geq 1-\frac{1}{\sqrt{n}}$, then if a symmetric strategy $\sigma$ satisfies $\sigma_x> p_x +\frac{1}{n}$, there exists some position $x'$ with $\sigma_{x'}<p_{x'}-\frac{1}{n}$.
    It follows by Proposition~\ref{prop:mixed-lower-bound} that $\sigma$ cannot be an equilibrium.

    Consider a position $x$ with $p_x \geq 1-\frac{1}{\sqrt{n}}$, and let $\sigma_x>p_x+\frac{1}{n}$.
    Then there must exist some position $x'$ with 
    \begin{align*}
        \sigma_{x'} &\leq \frac{1 - \sigma_x}{|\X^*| - 1} \\
        &< \frac{\frac{1}{\sqrt{n}}-\frac{1}{n}}{|\X^*|-1} \\
        &= \frac{n-\sqrt{n}}{n\sqrt{n}(|\X^*|-1)} \\
        %&\leq p_{\min}-\frac{1}{n} \\
        &\leq p_{x'} - \frac{1}{n},
    \end{align*}
    where the third inequality follows from the existence of some $x$ with $p_{x'}\geq\frac{1}{n}+\frac{n-\sqrt{n}}{n\sqrt{n}(|\X^*|-1)}$, a condition that is automatically satisfied given the condition on $n$ from Lemma \ref{lem:mixed-extreme}.
    Taking $n=\left(\frac{32}{\pmin}\right)\log\left(\frac{8}{\pmin}\right)$, plug into the bound on $p_{x'}$; $\frac{1}{n}+\frac{n-\sqrt{n}}{n\sqrt{n}(|\X^*|-1)}<\pmin$, so the bound holds.
    Thus $\sigma_{x'}< p_{x'} -\frac{1}{n}$, contradicting Proposition \ref{prop:mixed-lower-bound}.
    It follows that $\sigma_x\leq p_x+\frac{1}{n}$ for $p_x \geq 1-\frac{1}{\sqrt{n}}$.
    
\end{proof}

\section{Additional application results omitted from \S~\ref{sec:app}}
Here, we apply our main results to additional settings captured by our model. 

\subsection{Spatial voting}
In the spatial voting setting, the pseudo-target set $\X^*$ is the set of optimal ideologies for the voters in $\Y$ (that is, $d(x^*(y),y) = 0$). 
We can thus refer to each $x \in \X^*$ as a voter's ``ideal position.''
It follows a strategy profile is extreme if candidates only choose ideological positions that correspond to some voter's ideal position.

\begin{corollary}{(Spatial voting)}
    Consider the plurality rule with randomized tie-breaking.
    If $n \geq \frac{2}{\pmin}$, a pure Nash equilibrium exists that is extreme; moreover, the equilibrium covers the set of voters' ideal positions, i.e. for each voter $y$, at least one candidate chooses position $\arg\min_{x'}d(x', y)$.
    If $n > \frac{1}{\pmin}$, then all pure equilibria are extreme and cover all voters' ideal positions. 
    Moreover, for $n > \frac{1}{c} \geq \frac{2}{\pmin}$, the KL-divergence between the empirical distribution of candidate positions and the true distribution of voters' ideal positions is upper bounded by $\log \left( \frac{\lfloor cn\rfloor + 1}{\lfloor cn \rfloor} \right).$
    
    If $n > \max \left\{ 43, 8 \left(\frac{4}{\pmin} \log \frac{1}{\pmin} \right) \right\}$, any symmetric mixed Nash equilibrium $\sigma$ is extreme, and for each ideal position $x \in \X^*$, $|\sigma(x) - P(x)| \leq \frac{1}{n}.$
\end{corollary}

In other words, as the number of candidates grows much larger than the number of voters, those candidates' projected ideological positions will mimic the underlying distribution of voters' ideal positions. 
Instead of ideological convergence to the median voter, this result implies that candidates will commit to matching individual voters' ideologies to capture their votes. 
However, it is important to note that in most traditional political elections the number of voters is in fact much larger than the number of candidates; our results do not hold in these settings.

\subsection{Discrete Voronoi games}
In a discrete Voronoi game, each pseudo-target $x \in \X^*$ corresponds exactly to a user $v \in V$.
A strategy profile is extreme if influencers only place themselves on vertices.

\begin{corollary}{(Discrete Voronoi games)}
    If $n \geq \frac{2}{\pmin}$, a pure Nash equilibrium exists that is extreme; moreover, the equilibrium covers the set of users, i.e. for each user $v \in V$, at least one influencer chooses position $v$.
    If $n > \frac{1}{\pmin}$, then all pure equilibria are extreme and cover all users.
    Moreover, for $n > \frac{1}{c} \geq \frac{2}{\pmin}$, the KL-divergence between the empirical distribution of influencers and the true distribution of users over the vertices is upper bounded by $\log \left( \frac{\lfloor cn\rfloor + 1}{\lfloor cn \rfloor} \right).$
    
    If $n > \max \left\{ 43, 8 \left(\frac{4}{\pmin} \log \frac{1}{\pmin} \right) \right\},$ any symmetric mixed Nash equilibrium $\sigma$ is extreme, and for each user $v \in V$, $|\sigma(v) - P(v)| \leq \frac{1}{n}.$
\end{corollary}

That is, when the number of influencers in a network grows much larger than the number of users, they will strategically spread out in the network according to the underlying distribution of users in an attempt to maximize their corresponding Voronoi cell.

\end{document}